\documentclass[english,11pt,letterpaper]{article}
\usepackage[T1]{fontenc}
\usepackage[latin9]{inputenc}
\usepackage{geometry}
\usepackage{amsthm}
\usepackage{amsmath}
\usepackage{amssymb}
\usepackage{url}
\usepackage{hhline}
\usepackage{tablefootnote}
\usepackage{graphicx}
\usepackage{subcaption}

\makeatletter
\numberwithin{equation}{section}
\numberwithin{figure}{section}
\theoremstyle{plain}
\newtheorem{thm}{\protect\theoremname}
  \theoremstyle{definition}
  \newtheorem{defn}[thm]{\protect\definitionname}
  \theoremstyle{plain}
  \newtheorem{lem}[thm]{\protect\lemmaname}
  \theoremstyle{definition}
  
  \theoremstyle{remark}
  \newtheorem{rem}[thm]{\protect\remarkname}
  \theoremstyle{remark}
  \newtheorem*{rem*}{\protect\remarkname}

\@ifundefined{date}{}{\date{}}
\usepackage{tikz} 

 \usepackage[ruled,linesnumbered]{algorithm2e}

\input{Qcircuit}
\usepackage{graphicx,amsmath, amsthm, amssymb,color,url,booktabs,comment}  
\usepackage[colorlinks]{hyperref} 

\newcommand{\nc}{\newcommand}
\nc{\rnc}{\renewcommand}

\nc{\vev}[1]{\langle#1\rangle}
\nc{\grad}{{\vec{\nabla}}}

\DeclareMathOperator{\poly}{poly}
\DeclareMathOperator{\polylog}{polylog}
\DeclareMathOperator{\tr}{tr}
\DeclareMathOperator{\rank}{rank}

\newcommand{\be}{\begin{equation}}
\newcommand{\ee}{\end{equation}}
\newcommand{\bea}{\begin{eqnarray}}
\newcommand{\eea}{\end{eqnarray}}
\newcommand{\nn}{\nonumber}
\newcommand{\bi}{\begin{itemize}}
\newcommand{\ei}{\end{itemize}}
\newcommand{\bn}{\begin{enumerate}}
\newcommand{\en}{\end{enumerate}}
\def\beas#1\eeas{\begin{eqnarray*}#1\end{eqnarray*}}
\def\ba#1\ea{\begin{align}#1\end{align}}
\nc{\bas}{\[\begin{aligned}}
\nc{\eas}{\end{aligned}\]}
\nc{\bpm}{\begin{pmatrix}}
\nc{\epm}{\end{pmatrix}}

\def\nn{\nonumber}

\def\L{\left} 
\def\R{\right}

\newtheorem{cor}[thm]{Corollary}

\makeatletter
\newtheorem*{rep@theorem}{\rep@title}
\newcommand{\newreptheorem}[2]{%
\newenvironment{rep#1}[1]{%
 \def\rep@title{#2 \ref{##1}}%
 \begin{rep@theorem}}%
 {\end{rep@theorem}}}
\makeatother

\newreptheorem{thm}{Theorem}
\newreptheorem{lem}{Lemma}

\def\eps{\epsilon}

\DeclareMathOperator*{\E}{\mathbb{E}}

\def\benum{\begin{enumerate}}
\def\eenum{\end{enumerate}}
\def\bit{\begin{itemize}}
\def\eit{\end{itemize}}
\def\bdesc{\begin{description}}
\def\edesc{\end{description}}

\nc{\todo}[1]{\textcolor{red}{todo: #1}}

\def\begsub#1#2\endsub{\begin{subequations}\label{eq:#1}\begin{align}#2\end{align}\end{subequations}}
\nc\qand{\qquad\text{and}\qquad}
\nc\mnb[1]{\medskip\noindent{\bf #1}}
\nc\mn{\medskip\noindent}

\nc{\nle}{\nn \\ &=}  
\nc{\nnl}{\nn \\ &}  
\nc{\fot}{\frac{1}{2}} 
\nc{\oo}[1]{\frac{1}{#1}} 
\newcommand{\ben}{\begin{enumerate}}
\newcommand{\een}{\end{enumerate}}
\nc{\mc}{\mathcal}
\nc{\beq}{\begin{equation}}
\nc{\eeq}{\end{equation}}
\nc{\norm}[1]{\L\| #1 \R\|}

\nc{\onenorm}[1]{\L\| #1 \R\|_1} 

\DeclareMathOperator*{\argmax}{arg\,max}

\nc{\Ra}{\Rightarrow}
\nc{\zo}{\{0,1\}}	

\makeatother

\usepackage{babel}
  
  \providecommand{\definitionname}{Definition}
  \providecommand{\lemmaname}{Lemma}
  \providecommand{\problemname}{Problem}
  \providecommand{\remarkname}{Remark}
\providecommand{\theoremname}{Theorem}

\usepackage{mathtools}

\DeclarePairedDelimiter\floor{\lfloor}{\rfloor}

\begin{document}

\title{Sample Efficient Algorithms for Learning Quantum Channels in PAC Model and the Approximate State Discrimination Problem}
\author{Kai-Min Chung\thanks{Institute of Information Science, Academia Sinica. \texttt{kmchung@iis.sinica.edu.tw } This research is partially supported by the 2016 Academia Sinica Career Development Award under Grant no. 23-17.} \  and Han-Hsuan Lin\thanks{Department of Computer Science, The University of Texas at Austin \texttt{linhh@cs.utexas.edu}}}
\maketitle

\begin{abstract}
The probably approximately correct (PAC)  model~\cite{Valiant84atheory} is a well studied model in classical learning theory. Here, we generalize the PAC model from concepts of Boolean functions to quantum channels, introducing \emph{PAC model for learning quantum channels}, and give two sample efficient algorithms that are analogous to the classical ``Occam's razor'' result~\cite{blumer1987occam}.
The classical Occam's razor algorithm is done trivially by excluding any concepts not compatible with the input-output pairs one gets, but such an approach is not immediately possible with a concept class of quantum channels, because the outputs are unknown quantum states from the quantum channel. 

To study the quantum state learning problem associated with PAC learning quantum channels, we focus on the special case where the channels all have constant output. In this special case, learning the channels reduce to a problem of learning quantum states that is similar to the well known quantum state discrimination problem~\cite{discri-survey}, but with the extra twist that we allow $\eps$-trace-distance-error in the output. We call this problem \emph{Approximate State Discrimination}, which we believe is a natural problem that is of independent interest.

We give two algorithms for learning quantum channels in PAC model. The first algorithm has sample complexity  $$O\L(\frac{\log|C| + \log(1/ \delta)} { \eps^2}\R),$$ 
but only works when the outputs are pure states, where $C$ is the concept class, $\eps$ is the error of the output, and $\delta$ is the probability of failure of the algorithm. The second algorithm has sample complexity $$O\L(\frac{\log^3 |C|(\log |C|+\log(1/ \delta))} { \eps^2}\R),$$
and work for mixed state outputs. Some implications of our results are that we can PAC-learn a polynomial sized quantum circuit in polynomial samples, and approximate state discrimination can be solved in polynomial samples even when the size of the input set is exponential in the number of qubits, exponentially better than a naive state tomography.


\end{abstract}

\section{Introduction}

In computational learning theory, the Probably Approximately Correct (PAC) model of Valiant~\cite{Valiant84atheory} gives a complexity-theoretic foundation of what it means for a concept class to be (efficiently) learnable. In the most basic setting of PAC learning model, we want to learn a set of Boolean functions, $C = \{c:\zo^n \rightarrow \zo \}$, called the concept class. The goal of a learning algorithm $A$ is to  guess the identity of an unknown target concept $c^* \in C$ from samples $\{(x_1,c^*(x_1)),(x_2,c^*(x_2)),\dots\}$, where $\{x_1,x_2, \dots\}$ are inputs randomly drawn from a distribution $D$ that is unknown to $A$. Specifically, with error parameters $\eps$ and $\delta$, for all concept $c^* \in C$ and probability distribution $D$, $A$ is required to, given access to the samples $\{(x_1,c^*(x_1)),(x_2,c^*(x_2)),\dots\}$, with probability $1-\delta$, come up with a hypothesis $h \in C$ that is $\eps$-close to $c^*$, i.e. $\Pr_{x\leftarrow D}[c(x) \neq h(x)] \leq \eps$. Such a learning algorithm is called a proper\footnote{Proper means that the hypothesis $h$ must be inside the concept class $C$, whereas an improper learner can output any $h$ as the hypothesis. All learners in this paper are proper, and we sometimes omit the term ``proper''.} $(\eps,\delta)$-PAC learner for the concept class $C$. Of course, we would like the learner $A$ to be as efficient as possible in terms of both sample complexity (i.e., the number of samples $A$ needs to access) and time complexity, and ideally, polynomial in the input length $n$ and the error parameters  $\eps^{-1}$ and $\log(1/\delta)$. Since its introduction in the 80's by Valiant, PAC learning theory has been deeply studied to characterize when efficient learning is or is not possible. 


Following Valient's PAC learning model on Boolean functions, generalization to different kinds of concept classes has been proposed, including Boolean functions on continuous spaces~\cite{blumer1989learnability}, probabilistic Boolean functions~\cite{kearns1994efficient,alon1997scale}, functions with $\{0,\dots, n\}$ outputs~\cite{natarajan1989learning,bendavid1995characterizations}, and real valued functions~\cite{bartlett1996fat}. 

With quantum computers coming closer and closer into reality, it is natural to generalize the PAC learning model to quantum channels, capturing the learnability of quantum circuits or devices that we might build in the near future. Note that quantum states has an inherent "unlearnability", as manifested by the no-cloning theorem and uncertainty principle. Therefore this study of learnability of quantum channels has an interesting interaction between classical learning theory and quantum information theory.




 Formally, we define the \emph{PAC learning model for quantum channels} as follows: Let the \emph{concept class }$C$ be a finite set of known $d_1$ to $d_2$ dimensional quantum channels.  We are trying to learn an unknown quantum channel, the \emph{target concept }$c^*\in C$. In order to do this, we are given \emph{samples} $\{(x_1,c^*(x_1)),(x_2,c^*(x_2)),\dots\}$, where $\{x_1,x_2, \dots\}$ are classical descriptions of the input quantum states to the quantum channel $c^*$ and $\{c^*(x_1),c^*(x_2),\dots\}$ are the corresponding quantum states outputted by $c^*$. The inputs are drawn from a distribution $D$ unknown to the learner. Because of the no-cloning theorem, it is hard to justify holding both the inputs and outputs as unknown quantum states,  so we assume that we have full classical description of the input state and keep the outputted states as unknown quantum states, meaning that we hold a copy of the quantum state $c^*(x_i)$ rather than the full classical description of it. A proper $(\eps,\delta)$-PAC learner for the concept class $C$ of quantum channels is a quantum algorithm that for all concepts $c^* \in C$ and distribution $D$, takes the description of $C$ and $T$ samples $\{(x_1,c^*(x_1)),(x_2,c^*(x_2)),\dots\,(x_T,c^*(x_T))\}$ as input\footnote{Note that $D$ is not part of the input and is unknown to the learner.} and with probability $1-\delta$,  outputs a \emph{hypothesis} $h \in C$ that is $\eps$-close to the target concept $c^*$, where the distance between two concepts $h,c^*$ depends on the input distribution $D$ and is defined as $\Delta(h,c^*)= \mathbb{E}_{x\in D} \L[ \Delta_{tr}(h(x) ,c^*(x)) \R]$, i.e. the expected trace distance between the outputs averaged over $D$. 
 
We gave two algorithms for learning quantum channels in PAC model that in a sense generalize the classical Occam's razor algorithm~\cite{blumer1987occam}.  In particular, our algorithms have 
$\poly\log$ sample complexity in the size of the concept class.
The first algorithm has sample complexity  $$O\L(\frac{\log|C|+\log(1/ \delta)} { \eps^2}\R),$$ but requires the outputs to be pure states. The second algorithm has sample complexity  $$O\L(\frac{\log^3 |C|(\log |C|+\log(1/ \delta))} { \eps^2}\R),$$ while  outputs can be mixed. 

The Occam's razor algorithm~\cite{blumer1987occam} is a classical PAC learner for any finite sized concept class $C$ with sample complexity $O(\log|C|)$. The idea of the algorithm is simple: keep taking samples, check which concepts in the concept class do not agree with the samples and exclude them. One can show that every time a sample is taken, a constant fraction of the concepts that are $\eps$-far away from the target concept will be excluded, so an $\eps$-close hypothesis can be found in $O(\log|C|)$ samples.

Although the Occam's razor algorithm is simple, generalizing it to our PAC model for quantum channels is troublesome. The main difference is that when learning quantum channels, the outputs from the target concept are copies of unknown (possibly high dimensional) quantum states. By the nature of quantum mechanics, if we just have a few copies of a high dimensional quantum state, we can only learn a tiny fraction of information contained in the quantum state. Since we don't really know what the outputted state is, we cannot simply ``exclude all channels that do not output this state.'' Instead, we need to carefully design the measurement we take on the outputted states, getting the information useful in distinguishing the quantum channels in our concept class. Note that the sample complexities of both of our algorithms do not depend on the dimension of the outputted states.



As a possible application of our result, our algorithms for learning quantum channels in PAC model can be viewed as a sample-efficient way to do quantum process tomography~\cite{mohseni2008quantum} when we know that the target quantum processes comes from a finite set and only care about being correct on average over an input distribution. For example, if we try to  PAC-learn a polynomial sized quantum circuit of $n$-qubits, since there are only $2^{\poly(n)}$ possible polynomial sized circuits, our result shows that we can learn it in $\poly(n)$ samples, an exponential improvement over a naive process tomography that has no restriction on concept class size and inputs.

Note that this work studies the sample complexity instead of time complexity of learning. 
Just like various other cases in theoretical computer science where the oracle-based complexity does not match the time complexity of a problem, sample complexity and time complexity of learning quantum channels in PAC model is unlikely to match.  In particular, Arunachalam et al.~\cite{arunachalam2019quantum} showed that there is no polynomial time algorithm for learning $\text{TC}^0$ or AC$^0$ circuit even knowing $D$ is uniform unless LWE can be solved in polynomial time by a quantum computer.

\subsection{Approximate State Discrimination and Related Problems}

As stated previously, the most challenging part of our algorithms is how to extract information from unknown outputted quantum states to distinguish the channels. We isolate and study this problem by focusing on the special case where the channels are ``constant,'' i.e. every channel in the concept class outputs a fixed  quantum state irrespective of the input\footnote{The outputs of different concepts are still different.}. Since the input does not matter, we don't need to write it down anymore, so the samples are just copies of the fixed unknown quantum state, and since a concept is fully specified by its unique output state, we might as well describe the concept class as a set of quantum states.  In this special case, learning quantum channels in PAC model becomes an interesting hybrid of quantum state discrimination~\cite{pgm-am,discri-lowerbound,discri-ashley,discri-survey,discri-hypo} and quantum state tomography~\cite{tomography-aram,tomography-ryan}, and we named it the \emph{approximate state discrimination }problem. The approximate state discrimination problem is formalized as follows: Let $S$ be a known finite set of $d$-dimensional density matrices. We want to learn an unknown target state $\sigma \in S$ using as few identical copies of $\sigma$ as possible. A quantum algorithm is an $(\eps,\delta)$-approximate discriminator of $S$ if, for all $\sigma \in S$, it takes the description of $S$ and $T$ copies of $\sigma$ as input and with probability $1-\delta$ outputs a state $\rho \in S$ with $\Delta_{tr}(\rho,\sigma) \le \eps$. This problem is called approximate state discrimination because it is the same as the state discrimination problem except that $\eps$-approximate answers are allowed. 


Since approximate state discrimination is a special case of PAC learning quantum channels\footnote{We choose not to write up stand-alone algorithms for the approximate state discrimination problem as it will be very similar to that of PAC learning quantum channels. However, the reader can read the analysis of our algorithms with constant output assumptions to easily get the intuition behind them.}, it can also be solved with  $$O\L(\frac{\log|S|+\log(1/ \delta)} { \eps^2}\R)$$ samples if  $S$ consists of pure states and $$O\L(\frac{\log^3 |S|(\log |S|+\log(1/ \delta))} { \eps^2}\R)$$ samples if  $S$ consists of mixed states. 

To the knowledge of the authors, the approximate state discrimination problem has not been studied in the literature. 

As the reader might have already noticed, there are many problem revolving around trying to learn something from multiple copies of a unknown quantum state $\sigma$, including the approximate state discrimination problem we just proposed. To avoid possible confusion and highlight the significance of approximate state discrimination, we list the definitions of those problems in Table~\ref{fig:learning_sigma},  sample complexities in Table~\ref{fig:sample_sigma}, and provide comparisons to the approximate state discrimination problem in the following paragraphs.

\begin{table}
    \centering
    \begin{center}
\begin{tabular}{ |c|c|c| } 
 \hline
 Problem & $\sigma$ is from & Output \\   \hhline{|=|=|=|}

 Approximate State Discrimination & Finite Set $S$ & $\rho \in S$ s.t. $\Delta_{tr}(\rho,\sigma) \le \eps$  \\ 
  \hline
 Quantum State Discrimination & Finite Set $S$ & $\sigma$ \\ 
 \hline
 Quantum State Tomography & Anywhere & $\rho$ s.t. $\Delta_{tr}(\rho,\sigma) \le \eps$  \\ 
 \hline
 Quantum Property Testing & Anywhere & Decision: $\sigma \in S$ or $\sigma$ is $\eps$-far from $S$  \\ 
 \hline
 Quantum State Certification & Anywhere & Decision: $\sigma =\rho$ or $\sigma$ is $\eps$-far from $\rho$  \\ 
 \hline
 Shadow Tomography~\cite{scott-shadow} & Anywhere & numbers $\{b_i\}$ s.t. $|b_i-\text{Tr}(\sigma E_i)|\leq \eps$ for all i\\ 
 \hline
 Pretty Good Tomography~\cite{scott-learning-06} & Anywhere & $\rho$ s.t. $\mathbb{E}_{E_i\sim D}|\text{Tr}(\rho E_i)-\text{Tr}(\sigma E_i)|\leq \eps$ \\ 
 \hline
\end{tabular}
\end{center}
    \caption{Comparison of Quantum State Learning Problems. Note that in the "pretty good tomography" problem, the learner cannot choose the measurements to apply on $\sigma$. In shadow tomography and pretty good tomography,  $\{E_i\}$ are known binary measurements.}
    \label{fig:learning_sigma}
\end{table}

\begin{table}
    \centering
    \begin{center}
\begin{tabular}{ |c|c|c| } 
 \hline
 Problem & Upper bound & Lower bound \\   \hhline{|=|=|=|}

 Approximate State Discrimination (pure) & $O\L(\frac{\log|S|+\log(1/ \delta)} { \eps^2}\R)$ & $\tilde{\Omega}((1-\delta)(\log |S|)/\eps^2)$  \\ 
  \hline
  Approximate State Discrimination & $O\L(\frac{\log^3 |S|(\log |S|+\log(1/ \delta))} { \eps^2}\R)$ & $\tilde{\Omega}((1-\delta)(\log |S|)/\eps^2)$   \\ 
  \hline
 Quantum State Discrimination & $\frac{2(log|S|-\log \delta)}{-\log F}$~\cite{pgm-aram} & $\frac{
 \log |S| +\log \delta}{\log(\lambda d)}~\cite{pgm-aram} $\\ 
 \hline
 Quantum State Tomography & $O\L(\frac{ rd} { \eps^2}\R)$~\cite{tomography-aram,tomography-ryan} & $\tilde{\Omega}\L(\frac{ rd} { \eps^2}\R)$~\cite{tomography-aram,tomography-ryan}  \\ 
 \hline
 Property Testing of Quantum State (pure) & $O\L(\frac{ \log |S|} { \eps^2}\R)$~\cite{or-lemma} & $\Omega\L(\frac{1}{\eps^2}\R)$~\cite{property-review} \\ 
 \hline
 Property Testing of Quantum State & N/A & $\Omega\L(\frac{d}{\eps^2}\R)$~\cite{state-cert}\\
 \hline
 Quantum State Certification (pure) & $O\L(\frac{1}{\eps^2}\R)$~\cite{property-review} & $\Omega\L(\frac{1}{\eps^2}\R)$~\cite{property-review} \\ 
 \hline
 Quantum State Certification  & $O\L(\frac{d}{\eps^2}\R)$~\cite{state-cert} & $\Omega\L(\frac{d}{\eps^2}\R)$~\cite{state-cert} \\ 
 \hline
 Shadow Tomography~\cite{scott-shadow} & $\tilde{O}\L(\frac{\log(1/\delta)\log^4 M \log d}{\eps^4}\R)$ & $\Omega\L(\frac{min\{d^2, \log M\}}{\eps^2}\R)$ \\ 
 \hline
 Pretty Good Tomography~\cite{scott-learning-06} & $\tilde{O}\L(\frac{\log d}{\eps^8 \delta^4}\R)$ & $\Omega \L(\frac{\log d}{\eps \delta^4}\R)$  \\ 
 \hline
\end{tabular}
\end{center}
    \caption{Sample complexities of Quantum State Learning Problems. $\sigma$ is a $d$-dimensional quantum state and $\delta$ is the failure probability of the learner. If $\delta$ is not specified, the failure probability is $1/3$. Some problems have different complexity when $\sigma$ is promised to be a pure state. For quantum state tomography, $r$ is rank of $\sigma$. For quantum state tomography, $F$ is the upper bound on pair wise fidelity and $\lambda$ is the upper bound on eigenvalues of the density matrices. For shadow tomography, $M$ is the number of binary measurements. }
    \label{fig:sample_sigma}
\end{table}

In the \emph{quantum state discrimination} problem~\cite{pgm-aram,pgm-am,discri-lowerbound,discri-ashley,discri-survey,discri-hypo,barnett2009quantum,chefles2000quantum}, which is also called as quantum detection problem~\cite{discri-ashley} or quantum hypothesis testing~\cite{discri-hypo}, we are promised that the state $\sigma$  comes from a known finite set $S$, and we want to find out which $\sigma$ we were given exactly.\footnote{Because quantum state discrimination is studied by experts of many different backgrounds, it also have many different formulations. A commonly studied setting in the literature is the one-shot and average error case, where there is a known distribution $P$ over $S$, the learner is given one copy of $\sigma$ according to $P$ and is trying to minimize his error probability.} Since we cannot distinguish two quantum states that are arbitrary close to each other, the input set $S$ usually comes with a promise of minimum distance\footnote{Usually expressed in terms of maximum fidelity in the literature.} between its elements. Compared to the state discrimination problem, approximate state discrimination has the same promise of the finite input set, but allows an approximate output instead of finding the exact  answer. Because an approximate output is allowed, solving approximate state discrimination does not require the promise of a minimum distance between input states in $S$. Because we allow arbitrarily close input states, the naive state discrimination algorithm of taking several copies to amplify the minimum distance then taking a PGM (pretty good measurement) does not work for the approximate state discrimination problem, since the error probability of PGM is not bounded when some of the states are close to each other. In fact, there are several pathological distance structures that rules out some naive modifications to the PGM algorithm, for example, a "string" of states that are all close to their neighbors but with the states on the two heads far away from each other. See Figure~\ref{fig:string}. Solving the approximate state discrimination problem against \emph{any} distance structure is the main technical contribution of our work. 

In the problem of \emph{quantum state tomography}~\cite{tomography-aram,tomography-ryan} , we want to get an $\eps$-approximation of the unknown state  $\sigma$. Compared to quantum state tomography, approximate state discrimination has the same goal of finding an $\eps$-close output, but we have the extra promise that the unknown state comes from a known finite set $S$. As a result,  we get algorithms of sample complexity $O(\polylog |S|)$, which is \emph{independent} of $d$, the dimension of the target state.  
Compared to the $\Theta(d^2/\eps^2)$ sample complexity of quantum state tomography, we get an exponential improvement even with fairly large $|S|$.
For example, we get an exponential improvement when $|S|$ is exponential in the number of qubits, $\log d$
, in which case the approximate state discrimination problem can be solved in $O(\polylog(d))$ samples, while a full state tomography need  $\Theta(d^2/\eps^2)$ samples.

\begin{figure}
    \centering
    \includegraphics[trim={0 0 0cm 0},clip,width=0.6\textwidth]{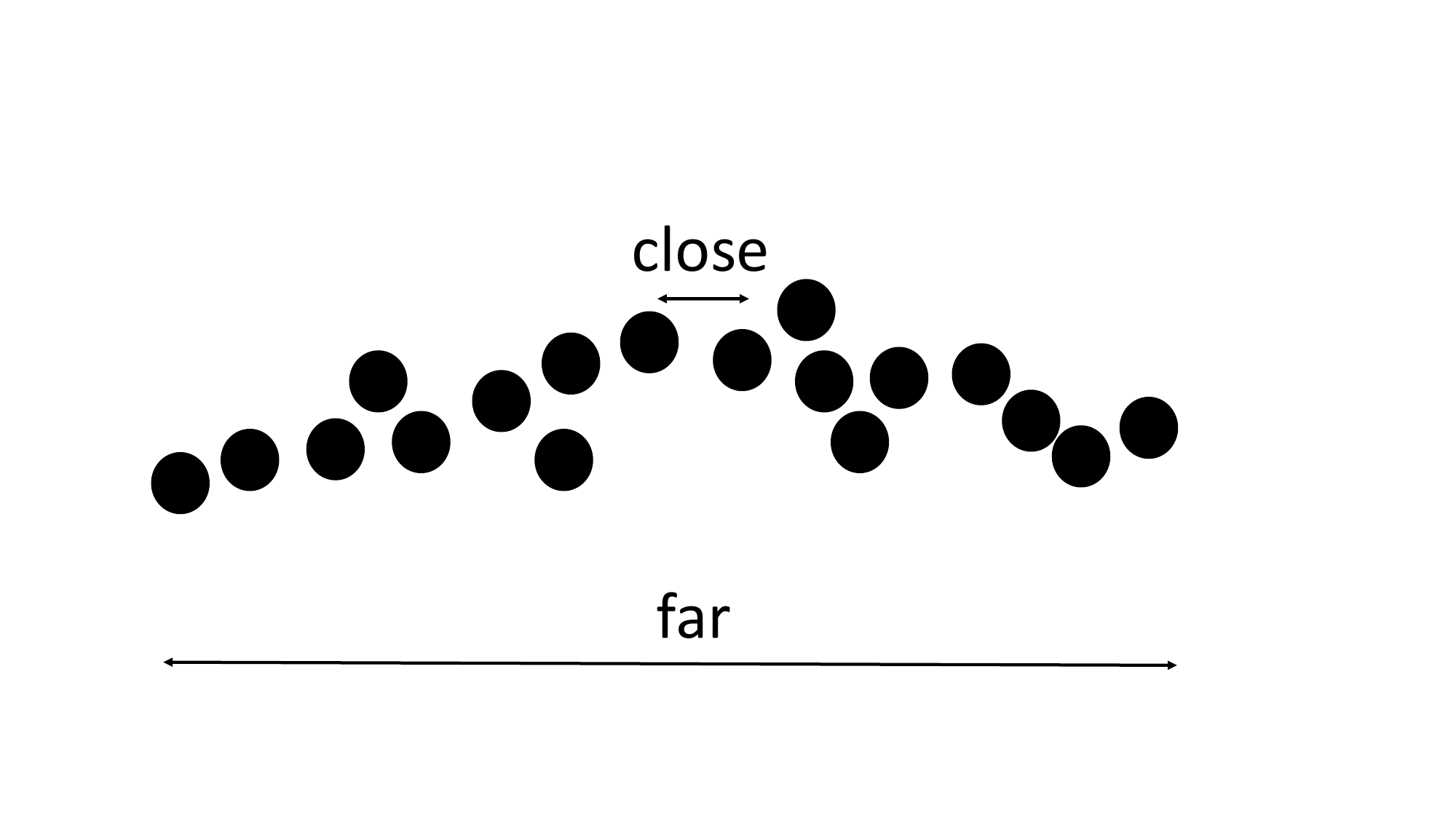}
    \caption{A pathological distribution of $S$ where every state is close to its neighbor but the states at two ends are far away from each other.}
    \label{fig:string}
\end{figure}

In the \emph{property testing of quantum state} problem~\cite{property-review,or-lemma},  we are given copies of an unknown quantum state $\sigma$ and want to determine whether $\sigma \in S$, where $S$ is a known (possibly infinite) set of quantum states, or $\sigma$ is $\eps$-far from anything in $S$. Harrow, Lin, and Montanaro~\cite{or-lemma} give an $O(\log|S|/\eps^2)$ upper bound on the \emph{property testing of quantum state} problem in the special case where $S$ is a finite set of \emph{pure} states. Comparing to~\cite{or-lemma}, our pure state algorithm has essentially the same $O(\log|S|/\eps^2)$ sample complexity. Note that in quantum property testing, the unknown state $\sigma$ does not always come from $S$, and we only want a decision answer instead of finding a state, so it is pretty different from approximate state discrimination.  Also note that the quantum property testing result of~\cite{or-lemma} cannot be generalized to arbitrary mixed states, as~\cite{state-cert} shows that to \emph{certify} a mixed state requires $\Omega({d/\eps^2})$ samples. In the quantum state certification problem, we are given copies of an unknown quantum state $\sigma$ and ask whether $\sigma$ is equal to some known state $\rho$ or $\eps$ far from it, so it is obviously a special case of quantum property testing with $|S|=1$, and the lower bound of $\Omega(d/\eps^2)$ is much larger than the sample complexity of  $O(\log |S|/\eps^2)$ in~\cite{or-lemma} unless $|S|$ is exponential or more in $d$.

Aaronson proposed the problem of pretty good tomography~\cite{scott-learning-06} and shadow tomography~\cite{scott-shadow}. Like approximate state discrimination, pretty good tomography and shadow tomography are ``speedup'' version of quantum state tomography. Unlike approximate state discrimination, they try to relax the output instead of restrict the input $\sigma$. Those problems aim to get good approximation to some known binary measurements $\{E_i\}$. In pretty good tomography, $E_i$ are drawn from an unknown distribution $D$ and the learner get samples that are pairs of ($E_i$ drawn, the corresponding (random) measurement result on $\sigma$). Note that the learner cannot do arbitrary measurements on $\sigma$. The learner aims to minimize the average error $\mathbb{E}_{E_i\sim D}|\text{Tr}(\rho E_i)-\text{Tr}(\sigma E_i)|$ over the unknown distribution $D$. In shadow tomography, the learner can do whatever he want, and he tries to estimate the measurement probability $\text{Tr}(E_i \sigma )$ for all $i$, minimizing the worst case error. Note that a shadow tomography learner do not need to output a quantum state that fit the list of probability he estimated.

\subsection{Related Works and Independent Work}

There are several works in the literature that study the sample complexity of PAC learning with different ways of generalization to quantum information.  Cheng,  Hsieh, and Yeh~\cite{min-hsiu-learning} studies the sample complexity of PAC learning arbitrary two outcome measurements, where the inputs are quantum states, and the learner has complete classical description of them. They show an upper of sample complexity linear in the dimension of the Hilbert space. Note that one can trivially get a lower bound of similar order by noticing that Boolean functions is a subset of two outcome measurements. Arunachalam
and de Wolf~\cite{wolf-learning} studies the sample complexity of PAC learning classical functions with quantum samples and shows that there is no quantum speed up. See~\cite{wolf-learning-survey} for a  survey of quantum learning theory.

\paragraph{Independent Work} ~\label{sec:indep work}

Independent to our work, in \cite{buadescu2020improved}, B{\u{a}}descu and O'Donnell formulate the problem of \emph{quantum hypothesis selection}. Quantum hypothesis selection can be viewed as a generalization of our approximate state discrimination problem where the unknown state $\sigma$ might not be in the hypothesis set $|S|$, and the learner what to find the state in $|S|$ that is closest  to the unknown state $\sigma$ (see Theorem 1.5 of ~\cite{buadescu2020improved} for the formal definition). This is similar to the agnostic learning model~\cite{Haussler92,KSS94}. Let $\eta$ be the minimum distance from the unknown state to something in $|S|$, B{\u{a}}descu and O'Donnell give an algorithm that finds some $\rho \in S$ such that $\Delta_{tr}(\rho,\sigma)\leq 3.01\eta +\eps$ using $O\L( \frac{\log^3 |S| +\log(1/\delta)}{\eps^2}\R)$ samples. Since quantum hypothesis selection is a generalization of approximate state discrimination, B{\u{a}}descu and O'Donnell's algorithm supersedes our algorithm for approximate state discrimination for the mixed state. 

However, it is important to note that B{\u{a}}descu and O'Donnell's algorithm requires many identical copies of the unknown state and thus does not generalize to our main result of PAC learning of quantum channels because every channel output might be a different state. On the other hand, as will shown in the following technical overview, our approach for approximate state discrimination involves a binary search through gap amplification and pretty good measurement and generalizes naturally to the PAC learning of quantum channels.


In ~\cite{aharonov2021quantum}, Aharonov, Cotler, and Qi introduced the notion of \emph{quantum algorithmic measurement}, which broadly captures the query and computational complexity of quantum experiments, including those that generate unknown identical quantum states. In \cite{huang2021information}, Huang, Kueng, and Preskill compared the complexity of classically or quantumly training a machine learning model for predicting outcomes of physical experiments.

\subsection{Technical Overview}\label{sec:intui}
The intuition behind both of our learning algorithms start with looking at the tensor product of all outputted states. The fidelity between such tensor produces decays exponentially in the number of samples drawn, so with enough samples , the tensor products from $\eps$-far concepts will become almost orthogonal (see Lemma~\ref{lem:concept-amp-state}),   so intuitively, we should be able to distinguish between them. 


\paragraph{Pure State algorithm} In the case where the channels always output pure states, we have a rather simple algorithm. The key part is a theorem by Sen~\cite{pranab} on high dimensional random orthonormal measurements, which states that if we do a  measurement of random orthonormal basis on two pure states, with high probability\footnote{The probability goes to 1 as the dimension goes to infinity.}, the trace distance between the distribution of measurement outcome is lower bounded by a constant times the trace distance between those two states (see Theorem~\ref{lem:random-measure}). This result might seem counter-intuitive, but remember that a random orthonormal measurement in $d$ dimension has $d$ possible outputs instead of $2$. With this theorem in hand, the algorithm is rather easy: take enough samples to amplify the distance between outputted states and  do a random orthonormal measurement on each sample. Choose the hypothesis as the channel that most likely to give the measurement result. 


\paragraph{Mixed State algorithm} Our thought process on designing a learner for the channels that output mixed states is the following. In this case,   Theorem~\ref{lem:random-measure} does not give us a useful result, so we need to find something else. Noticing the connection to the quantum state discrimination problem, we turned to pretty good measurement (PGM, Definition~\ref{def:pgm}), a well studied tool for solving the quantum state discrimination problem. However, the lack of minimum distance between our outputted states is pretty pathological to PGM, so it was pretty easy to self-reject all our attempts. Following that, we sought guidance from the analysis of classical Occam's razor algorithm, where a constant fraction of concepts are ruled out by each sample. We tried to divide the concept class into two sets, then do a PGM to distinguish those two, so we can recurse this into a binary search. Cutting the concepts into two sets does not work either because there can be concepts really close to any cut, which again is pathological to PGM. At this point, we realized that we need to have some kind of minimum distance for our PGM, so we cut the concept class into \emph{three} sets, $S_{yes}$, $S_{no}$, and $S_{unknown}$. We set a minimum distance $\gamma$ between elements of $S_{yes}$ and $S_{no}$, so those two sets can be distinguished. This is the idea that works out. See Figure~\ref{fig:cuts} for a graphical representation. 

\begin{figure}
     \centering
     \begin{subfigure}{0.45\textwidth}
  \centering
  \includegraphics[width=0.9\linewidth]{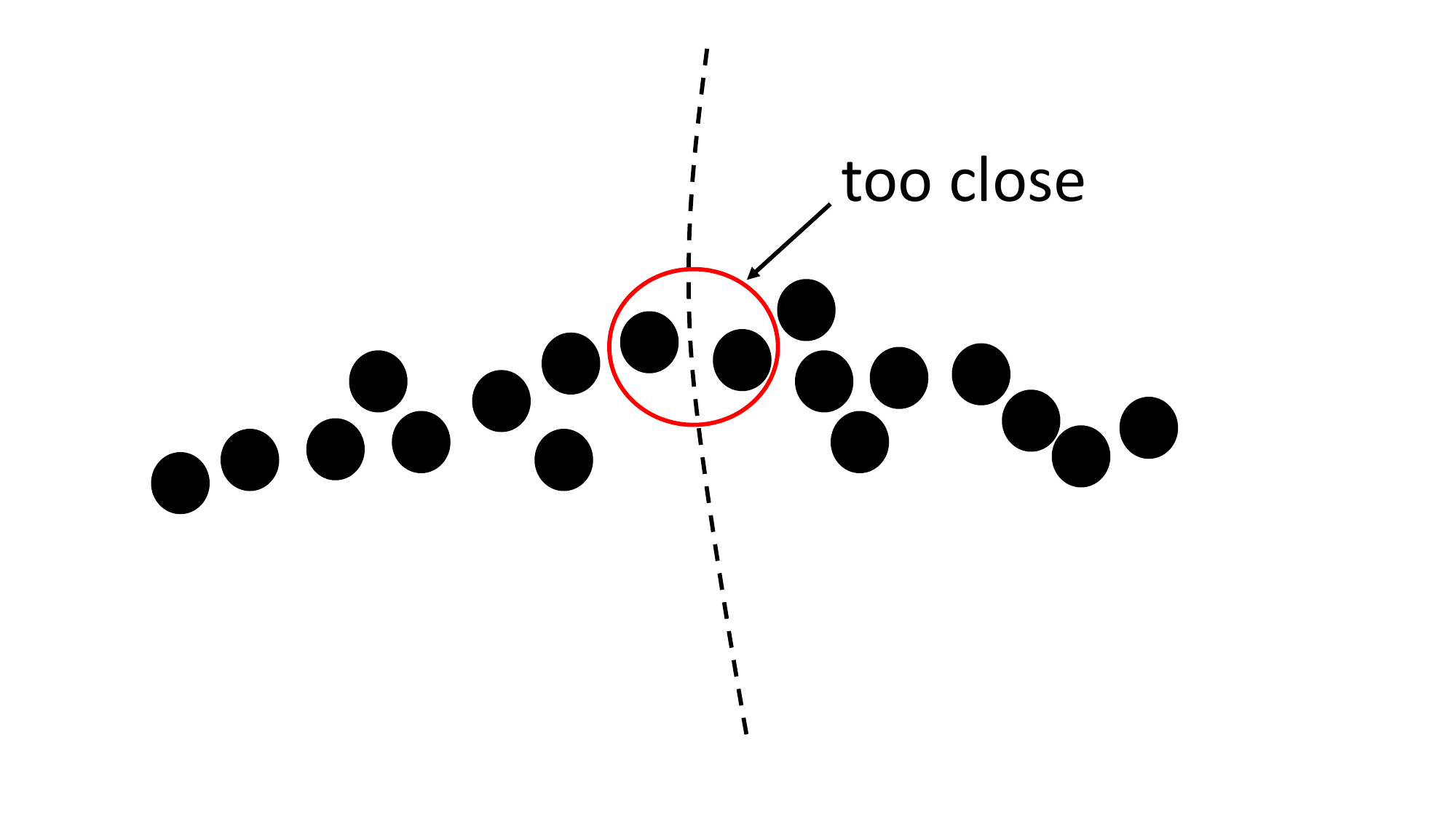}
  \caption{}
  \label{fig:1cuts}
\end{subfigure}%
\begin{subfigure}{.45\textwidth}
  \centering
  \includegraphics[width=0.9\linewidth]{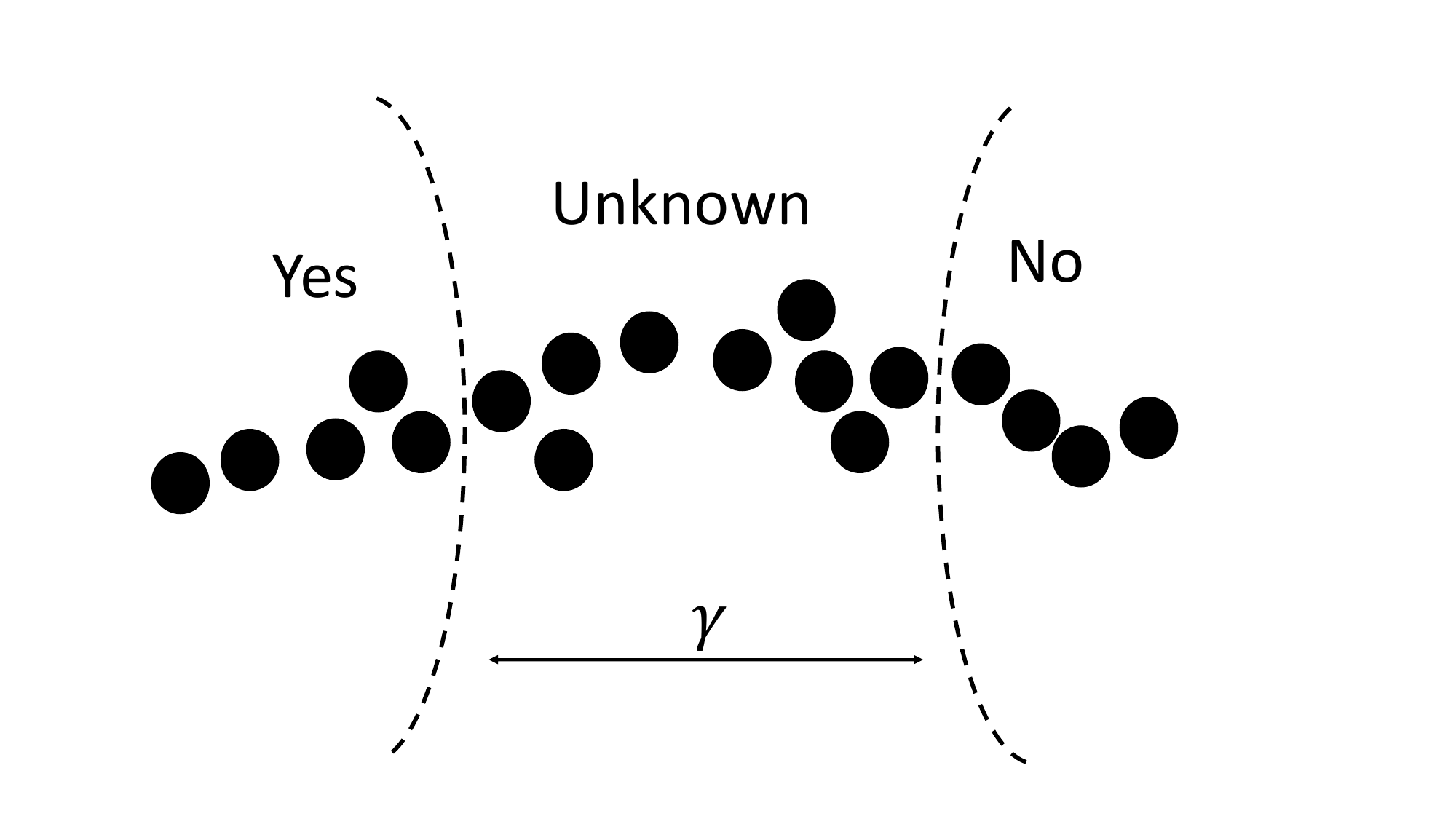}
  \caption{}
  \label{fig:2cuts}
\end{subfigure}%
     \caption{(a) Pathological case when trying to cut the concept class into two sets. (b) Cutting the concept class into three sets.}
     \label{fig:cuts}
 \end{figure}

To follow our intuition in the previous paragraph, we give a definition about the distance between two sets of quantum states. Actually fidelity is more useful than trace distance, so we give the following definition of fidelity between sets of quantum states, which is the maximum fidelity among all pairs: $$F\L( S_{yes},S_{no} \R) = \max \L\{F(\sigma,\rho)|\sigma\in S_{yes}, \rho \in S_{no} \R\}$$

\paragraph{Bichromatic State Discrimination Problem (BSD)}
The key component our mixed state algorithm is solving what we called \emph{$(\eta,N)$-Bichromatic State Discrimination Problem} (BSD). The $(\eta,N)$-Bichromatic State Discrimination Problem is defined as follows: given complete information of two sets of quantum states, $S_{yes}$ and $S_{no}$, with fidelity $F(S_{yes}, S_{no})\leq \eta$ and size $ S_{yes} \leq N$, $ S_{no} \leq N$, and one copy of an unknown quantum state $\sigma$, the goal is to decide whether $\sigma \in S_{yes}$ or $\sigma \in S_{no}$. A quantum algorithm solves $(\eta,N)$-BSD with error $\delta$ if for all $S_{yes}$ and $S_{no}$ such that $F(S_{yes}, S_{no})\leq \eta$, $ S_{yes} \leq N$, and $ S_{no} \leq N$, given complete information about $S_{yes}$ and $S_{no}$ and one copy of an unknown quantum state $\sigma$ as input to the algorithm,  the algorithm output a $yes/no$ answer satisfies the following two conditions\footnote{The learner can output anything if $\sigma$ does not come from either of the two sets.}:
\begin{enumerate}
    \item  If $\sigma \in S_{yes}$, the learner outputs $yes$ with probability $(1-\delta)$.
    \item If $\sigma \in S_{no}$, the learner outputs $no$ with probability $(1-\delta)$.
\end{enumerate}
See Figure~\ref{fig:bsd} for some graphical intuition of BSD.

\begin{figure}
    \centering
    \includegraphics[trim={0 0 0cm 0},clip,width=0.6\textwidth]{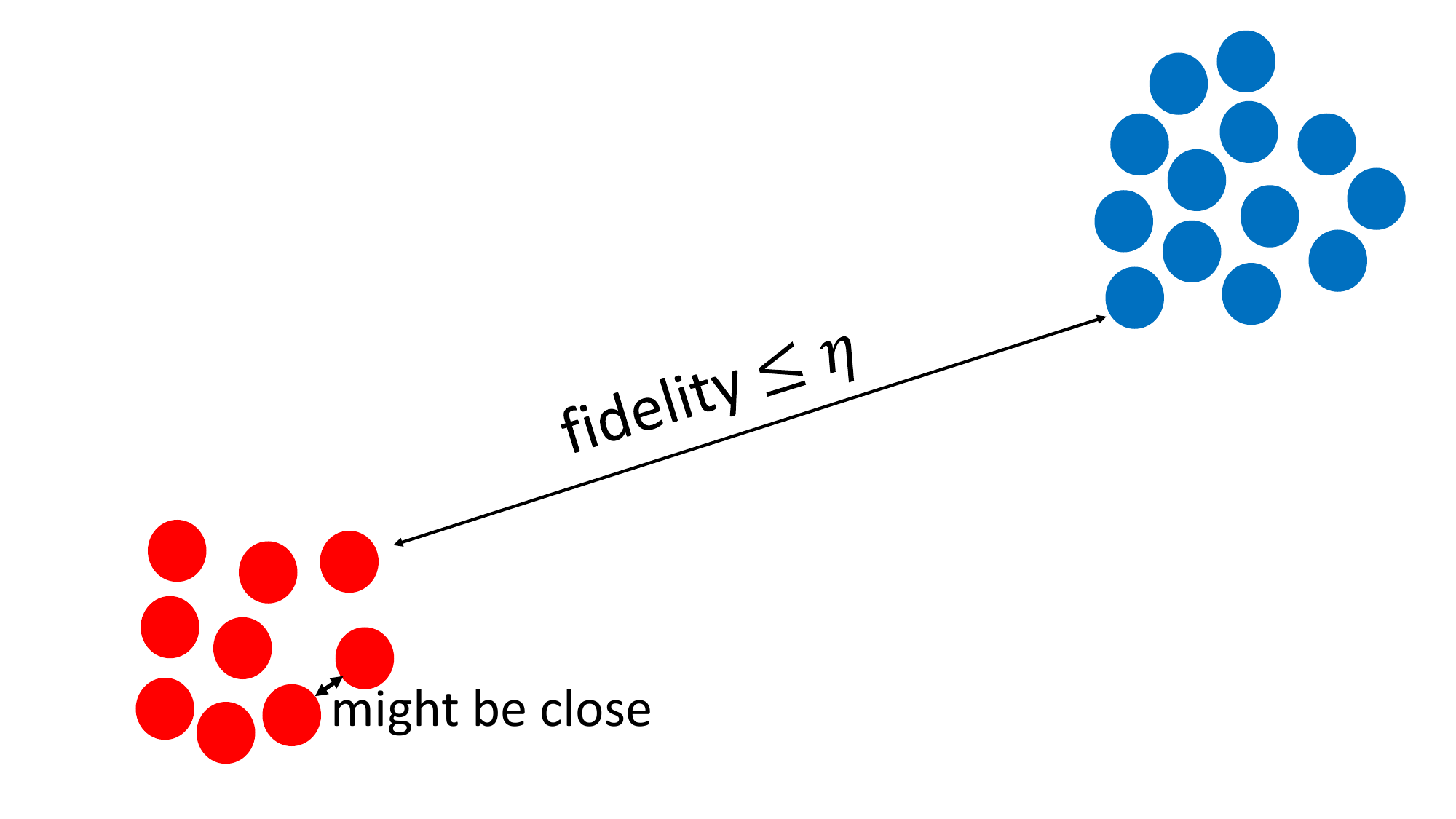}
    \caption{Bichromatic State Discrimination Problem}
    \label{fig:bsd}
\end{figure}

 Note that BSD only requires maximum fidelity \emph{between} the two sets; two states from the same set can be arbitrarily close. This does not violate quantum state discrimination lower bounds because the solver only needs to discriminate between the two sets.
 
 We are able to show that BSD can be solved with good enough parameter:
\begin{thm}\label{thm:bsd}
There exist an algorithm that solves $(\eta,N)$-BSD with error $\delta=N^2 \eta$.
\end{thm}

The proof idea of Theorem~\ref{thm:bsd} is trying to apply PGM on  $S_{yes}\cup S_{no}$.  We start with the observation that the result of  \cite{pgm-am} and \cite{pgm-bk}, which gives an upper bound on PGM's error probability of mistaking one state as other states, can be generalized to an upper bound on PGM's error probability of mistaking one subset of states to its complement subset (See Appendix~\ref{app:pgm}). This almost gives us the required error bound for BSD, except that the PGM result is for the average case, where $\sigma$ is drawn from some probability distribution, so we turned it into a worst case result with the minimax argument of \cite{pgm-aram}.

\paragraph{Back to Learning Quantum Channels} With BSD solved, we can get an algorithm that recursively exclude a constant fraction of the concept class. In each recursion, the algorithm partition the remain concepts into three sets, $S_{yes}, S_{unknown}$, and $S_{no}$. Ideally, $S_{yes}$ and $S_{no}$ both occupy a constant fraction of the remaining concepts and have minimum distance $\gamma=\Omega(1/\poly\log |C|)$. Noticing that the fidelity between tensor products of outputs decays exponentially with number of samples by lemma~\ref{lem:concept-amp-state}, the BSD between $O(\log |C|/\gamma)$ samples of $S_{yes}$ or $S_{no}$ can be solved with high probability. If the target concept is in $S_{yes}$, the BSD solver will return $yes$ with high probability, and if the target concept is in $S_{no}$, the BSD solver will return $no$. If the target concept is from $S_{unknown}$, the BSD solver might return anything, but what we can be sure is that, if the BSD solver returned $yes$, the target concept is not from $S_{no}$, and if the BSD solver return $no$, the target concept is not from $S_{yes}$. Therefore, we can always exclude either $S_{yes}$ or $S_{no}$ as possible target concept. 

There is another complication in that the distance between the concepts depends on the unknown distribution $D$ and thus cannot be calculated. In stead, we use the \emph{empirical distance}  between concepts, $\Delta_{emp}(c_1,c_2)=\frac{1}{T}\sum_{i=1}^T \L [ \Delta_{tr} \L(c_1(x_i),c_2(x_i)\R) \R ] $, where $\{x_i\}$ are the inputs points we drawn in each recursion. Our calculation shows that that the error incurred from this change of distance measure is negligible. 

\paragraph{Partition Sub-algorithm} It is not always possible to have an ideal partition where $S_{yes}$ and $S_{no}$ are both constant-fraction sized\footnote{By "constant-fraction sized" we mean "occupies a constant fraction of the remaining concepts".} and separated by the gap $\gamma$. Therefore, we designed a classical partition sub-algorithm (Algorithm~\ref{alg:partition}) to handle these exceptions. 

An example where the ideal partition is not possible is the extreme case where every concept in the concept class is literally identical to each other. Note that in this extreme case can be trivially solved by output anything in the concept class as the hypothesis because everything is $\eps$-close to $c^*$. 

Our partition sub algorithm builds on the intuition of what happened in the above extreme case. More specifically, our partition algorithm will not reserve a constant-fraction sized $S_{no}$ if a significant fraction of $C$ is clustered around a concept. In such case, we choose the cluster as $S_{yes}$ with a $\gamma$-thick ``shell'' of $S_{unknown}$ around it. If we measured $no$, we can rule out $S_{yes}$, which is a constant fraction of $|C|$. If we measured yes, we can output the center of the cluster as the hypothesis, and we tune $\gamma$ so that everything in either $S_{yes}$ or $S_{unknown}$ is $\eps$-close to the center. This completes our algorithm for  mixed state outputs.







\subsection{Lower Bounds and Agnostic Model}\label{sec:LB}

We complement our positive results on the sample complexity of PAC learning quantum channels with two simple lower bounds. First, by adapting a lower bound argument in~\cite{tomography-aram}, we prove that $\tilde{\Omega}((\log |C|)/\eps^2)$ samples are necessary to PAC learn quantum channels when the outputs are pure states, showing that our positive result is tight in the dependency on $|C|$ and $\eps$. In particular, for the dependency on $\eps$, this is in contrast with the classical results on the sample complexity for PAC learning concepts with Boolean outputs, where a tight $\Theta((\log|C|)/\eps)$ sample complexity is known~\cite{Kearns:1994:ICL:200548,hanneke2016optimal}.\footnote{The classical results show that the sample complexity is characterized by the VC dimension of the concept class $C$. In the case that $C$ is finite, $\log|C|$ is a trivial upper bound on the VC dimension.}

Agnostic model is a learning model closely related to the PAC model, and the two models have similar sample complexity~\cite{Haussler92,KSS94}. In the agnostic model, the samples comes from a concept $c_s$ that is not necessarily inside the concept class $C$. Accordingly, the goal of the learner is to find, with $\eps$-distance error, the target concept $c^*\in C$ that is closest to $c_s$. We introduce the agnostic model for learning quantum channels, see section~\ref{sec:agno} for details. Interestingly, in stark contrast to our algorithms that have dimension-independent sample complexity for learning quantum channels in PAC model, we found an $\Omega(\sqrt{d})$ lower bound on the sample complexity for learning quantum channels in agnostic model with output dimension $d$. Thus, in the agnostic model, learning quantum channels requires number of samples polynomial in the dimension, so it is not possible to efficiently learn quantum channels with large output dimension.  Also, our negative example is in fact classical in nature, consisting of two concepts that output classical distributions, so learning classical distributions efficiently in agnostic model in large dimension is also impossible. However, since quantum pure states are not generalizations of classical distributions, the possibility of sample efficiently learn quantum channels with \emph{pure state output} in agnostic is still open.

\subsection{VC dimension?}
The sample complexity of learning classical Boolean functions are tightly characterized by the \emph{VC dimension}~\cite{vapnik2015uniform,blumer1989learnability}. Ideally, we would like to generalize VC dimension to our PAC model for quantum channels and tightly characterize the sample complexity of any channel. However, this goal seems to be too ambitious. {We have two reasons to believe that generalizing VC dimension to quantum channels might be too hard. First note that the quantum channel might output mixed states, which generalize classical distributions over exponentially many discrete values. To our knowledge, VC dimension for classical functions that output such a large distribution has not been studied, so a VC dimension for quantum channels would be two steps ahead. Second, classically the sample complexity in the agnostic model is also characterized by the VC dimension. But as we shown in Section~\ref{sec:LB}, PAC model and agnostic model have very different sample complexity on learning channels. Therefore, we instead seek an analog of the Occam's razor result as a first step to study our PAC model. 

\section{Preliminary}


Throughout this paper, $\log$ is base 2 and $\ln$ is base $e$.

We use $\norm{\cdot}_1$ to denote the trace norm $\norm{A}_1=\tr \sqrt{A^\dag A}$. We use $\norm{\cdot}_2$ or $\norm{\cdot}_F$ to denote the Frobenius norm $\norm{A}_2=\sqrt{\tr (A^\dag A)}$.

Denote the trace distance and fidelity between two distribution $D_1,D_2$ as $\Delta_{tr}\L(D_1,D_2 \R)$ and $F(D_1,D_2)$, where the trace distance is equal to the total variation distance. Denote the trace distance and fidelity between two quantum states $\rho_1,\rho_2$ as $\Delta_{tr}(\rho_1,\rho_2)=\fot \norm{\rho_1-\rho_2}_1$ and $F(\rho_1,\rho_2)=\norm{\sqrt{\rho_1}\sqrt{\rho_2}}_1$. For a quantum state $\sigma$ and a quantum measurement $M$, denote $M(\sigma)$ as the output probability distribution when applying $M$ on $\sigma$.

Note that fidelity and trace distance are related by 
$$ 1-F \leq \Delta_{tr} \leq \sqrt{1-F^2}.$$

For two quantum channel concepts $c_1,\, c_2$, define the distance between them with respect to $D$ as $$\Delta(c_1,c_2)= \mathbb{E}_{x\in D} \L[ \Delta_{tr}(c_1(x),c_2(x)) \R].$$ We say that $c_1,\, c_2$ are $\eps$-close if $\Delta(c_1,c_2) \leq \eps$ and $\eps$-far if $\Delta(c_1,c_2) \geq \eps$.  For two sets of concepts $S_1$ and $S_2$, define the distance between them as $\Delta\L( S_1,S_2 \R) = \min \L\{\Delta(c_1,c_2)|c_1\in S_1, c_2\in S_2 \R\}$. 
\subsection{Chernoff Bound}

We use the following standard multiplicative version of Chernoff bound.

\begin{thm} Let $X_1,\dots,X_T \in [0,1]$ be independent random variables with $\E[X_i] = \mu_i$. Let $X = (1/T) \sum_i X_i$, $\mu = (1/T) \sum_i \mu_i$ and $\alpha \in (0, 1)$. We have 
$$ \Pr[ |X - \mu | \geq \alpha \mu] \leq 2^{-\Omega(\alpha^2 T \mu )}.$$
\end{thm}

\subsection{Pretty Good Measurement}
The pretty good measurement (PGM) is defined as follows:
\begin{defn}[pretty good measurement]\label{def:pgm}
Let $\{\sigma_i\}$ be a set of density matrices and $ \{p_i\}$ a probability distribution over $\{\sigma_i\}$. Define 
\ba
A_i=p_i \sigma_i,\, A=\sum_i A_i.
\ea
The PGM associated with $\{\sigma_i\}, \{p_i\}$ is the measurement $\{E_i\}$ with 
\ba
E_i = A^{-1/2} A_i A^{-1/2}.
\ea
\end{defn}

\section{Problem Definitions}
In this section we describe the PAC model of learning quantum channel and approximate state discrimination.

\subsection{Classical PAC Learning Model}
We start with a review of the classical PAC learning model. 

In the classical probably approximately correct (PAC) learning model, a learner tries to learn a \emph{target concept} $c^* \in C$ from a known \emph{concept class} $C$, which is a set of Boolean functions $c: \zo^n \rightarrow \zo$, with respect to an \emph{unknown} distribution $D$ over the input domain $\zo^n$. Specifically, the learner is given 
access to a sample oracle ${\cal O}_{c^*, D}$, which generates i.i.d. samples $(x_i, c^*(x_i))$, where each $x_i \leftarrow D$ is drawn according to the distribution $D$, and outputs a \emph{hypothesis} $h \in C$.\footnote{The requirement that the hypothesis $h$ is in the concept class $C$ is referred to as proper learning. We focus on proper learning since our algorithms satisfy this property.} 
The distance between two concepts $c$ and $h$ under the distribution $D$ is defined as $\Delta_{D}(c,h)=\mathbb{E}_{x\sim D} \L|c(x)-h(x)\R|$. The goal of the learner is to find a hypothesis $h$ with sufficiently small distance $\Delta_{D}(c^*,h)$ to $c^*$.

A learning algorithm $A$ is a  \emph{proper $(\eps,\delta)$-PAC learner} for a concept class $C$ if the following holds: For every $c^* \in C$ and distribution $D$, given oracle access to ${\cal O}_{c^*, D}$, $A^{{\cal O}_{c^*, D}}$ outputs an $h \in C$ such that $\Delta_{D}(c^*,h) \leq \eps$ with probability at least $1-\delta$. 
The sample complexity of $A$ is the maximum number of samples $T$ that $A$ needs to query ${\cal O}_{c^*, D}$ to output $h$. The  \emph{proper $(\eps,\delta)$-PAC sample complexity} of a concept class $C$ is the minimum sample complexity over all learners.
%
%
A  $\tilde{\Theta}((\log|C|)/\eps)$ sample complexity is known~\cite{Kearns:1994:ICL:200548,hanneke2016optimal}.\footnote{We use $\tilde{\Theta}$ to denote $\Theta$ with log factors. The classical results show that the sample complexity is $\Theta\L((d+\log  1/\delta)/\eps \R)$, where $d$ is the VC dimension of the concept class. In the case where $|C|$ is finite, $\log|C|$ is a trivial upper bound on $d$, and there are concept classes whose VC dimension $d$ matches $\log|C|$.}


\subsection{Learning Quantum Channels in PAC model}


We now generalize classical PAC learning to the context of learning quantum channels. As above, we consider a learner trying to learn a target concept $c^* \in C$ from a known concept class $C$ with respect to an unknown distribution $D$. Here, we consider the concept class $C$ as a finite set of known $d_1$ to $d_2$ dimensional quantum channels, and $D$ as a  distribution over the Hilbert space of dimension $d_1$.
Precisely, the learner is given 
access to a sample oracle ${\cal O}_{c^*, D}$ and outputs a hypothesis $h\in C$.
The oracle ${\cal O}_{c^*, D}$ generates i.i.d. samples $(x_i, c^*(x_i))$, where each $x_i \leftarrow D$ is the classical description of a state drawn according to the distribution $D$, and $c^*(x_i)$ is the (potentially mixed) quantum state outputted by $c^*$ on input $x_i$.

The distance between two concepts $c$ and $h$ under the distribution $D$ is the expected trace distance $\Delta(c, h)= \mathbb{E}_{x\in D} \L[ \Delta_{tr}(c(x) ,h(x))  \R]$. The goal of the learner is to find a hypothesis $h \in  C$ with sufficiently small  $\Delta(c^*,h)$.

 A quantum learning algorithm $A$ is a \emph{proper
 $(\eps,\delta)$-PAC learner} for  $C$ if the following holds: For every $c^* \in C$ and distribution $D$, given oracle access to ${\cal O}_{c^*, D}$, $A^{{\cal O}_{c^*, D}}$ outputs an $h\in C$ such that $\Delta_{D}(c^*,h) \leq \eps$ with probability at least $1-\delta$. The sample complexity of $A$ is the maximum number of samples $T$ that $A$ needs to query ${\cal O}_{c^*, D}$ to output $h$. The  \emph{proper $(\eps,\delta)$-PAC sample complexity} of a concept class $C$ is the minimum sample complexity over all learners.

\subsection{Approximate State Discrimination}
 Let $S$ be a finite set of $d$-dimensional density matrices. We want to learn a target state $\sigma \in S$ using as few identical copies of $\sigma$ as possible. A quantum algorithm is an $(\eps,\delta)$-approximate discriminator of $S$ if it takes the description of $S$ and $T$ copies of $\sigma$ as input and with probability $1-\delta$ outputs a state $\rho \in S$ with $\Delta_{tr}(\rho,\sigma) \le \eps$, for any $\sigma \in S$. 

Note that approximate state discrimination can be viewed as a special case of PAC learning quantum channels with constant output, so the algorithms for PAC learning quantum channels in Section~\ref{sec:pure} and Section~\ref{sec:mixed} trivially works for approximate state discrimination.

\section{PAC Learning Quantum Channels with Pure State Output}\label{sec:pure}

The algorithm follows ideas by Sen~\cite{pranab}, who shows that random orthonormal measurement preserves trace distance between pure states. One can then apply random orthonormal measurements on each sampled output and take enough samples to amplify the distance between $\eps$-far concepts to  $1- O({1}/{ | C |})$ and show that the probability for the maximum likelihood estimate to select a $\eps$-far concept over the target concept is less than  $O({1}/{ | C |})$. Take a union bound and we have a bounded error probability.

\begin{thm} \label{thm:pure}
Algorithm \ref{alg:pure} is a proper $(\eps, \delta)$-PAC learner for any concept class $C$ of  quantum channels with pure state outputs, using $$O\L(\frac{(\log|C|)+\log(1/ \delta)} { \eps^2}\R)$$ samples.
\end{thm}

\begin{algorithm}[H]
  Take $T=\Theta((\log|C|+\log(1/\delta))/\eps^2)$  samples $(x_1,\sigma_1),(x_2,\sigma_2),\dots,(x_T,\sigma_T)$ \; 
  Do a random orthonormal measurement\footnote{The measurement has $d_2$ outcomes, where $d_2$ is the dimension of output quantum state.} $M_i$ on each output state $\sigma_i$. Let the measured outputs be $\{ z_i \}$ \; \nllabel{step-rand-meas} 
  Output the concept $h\in C$ that is most likely to give the measured result of line~\ref{step-rand-meas}: 
\ba
h=\argmax_{c \in C} \Pi_{i\in[T]} \Pr[M_i(c(x_i))=z_i] \nn
\ea 
\caption{algorithm for pure state output]}
\end{algorithm}\label{alg:pure}

We need the following theorem to prove the correctness of Algorithm \ref{alg:pure}. First we state the result 1 of~\cite{pranab} (lemma 4 of arxiv version): 
\begin{thm}[random orthonormal measurement~\cite{pranab}]\label{lem:random-measure}
Let $\sigma_1,\, \sigma_2$ be two density matrices in $\mathbb C^d$. Define $r:=\rank(\sigma_1-\sigma_2) $. There exists a universal constant $k>0$ such that if $r<k\sqrt{d}$  then with probability at least $1-\exp(-kd/r)$ over the choice of a random orthonormal measurement basis $M$ in $\mathbb{C}^d$, $\norm{M(\sigma_1)-M(\sigma_2)}_1 > k\norm{\sigma_1-\sigma_2}_F$. \footnote{Recall that $M(\sigma)$ is the output distribution of the measurement $M$ on state $\sigma$. }
\end{thm}
Note that if $\sigma_1,\, \sigma_2$ are pure states,  $r<2<k\sqrt{n}$ for large enough $n$ and $\norm{\sigma_1-\sigma_2}_1\le \sqrt{2}\norm{\sigma_1-\sigma_2}_F$ so that  $\Delta_{tr}(M(\sigma_1),M(\sigma_2)) > k/\sqrt{2}\Delta_{tr}(\sigma_1,\sigma_2)$.

The following lemma shows how trace distance of the measured result grows when we take multiple samples.
\begin{lem} [trace distance amplification]\label{lem:trace-amp}
Let $X_1,X_2,\dots,X_T$ be $T$ independent distributions and so are $Y_1,Y_2,\dots,Y_T$. Denote the joint distribution  $(X_1,X_2,\dots,X_T)$ as $X$ and $(Y_1,Y_2,\dots,Y_T)$ as $Y$. Suppose that 
\ba
\sum_i \Delta_{tr}(X_i,Y_i) =T \eps ,
\ea
then
\ba
 \Delta_{tr}(X,Y) \geq 1 - 2^{-\Omega(T \eps^2)}
\ea
\end{lem}
\begin{proof}
By Cauchy-Schwarz inequality, 
\ba
\sum_i (\Delta_{tr}(X_i,Y_i))^2 \ge T \eps^2 ,
\ea
Then the joint fidelity is bounded by
\ba
F \L( X , Y \R) &=\Pi_i F \L( X_i,  Y_i \R) \nnl
\leq \Pi_i \sqrt{1-\L(\Delta_{tr}\L( X_i,Y_i\R) \R)^2} \nnl
\leq   \exp \L[-\fot \sum_i \L(\Delta_{tr}\L(X_i,Y_i \R) \R)^2\R] = 2^{-\Omega(T \eps^2)}, 
\ea
where the last inequality is true because  $1-x \leq e^{-x}$. And the joint  trace distance is 
\ba
 \Delta_{tr}\L(X,Y \R) \ge 1- F \L( X,Y\R) =1 - 2^{-\Omega(T \eps^2)}.
\ea
\end{proof}

The following lemma analyzes the effectiveness of maximum likelihood estimate.
\begin{lem}\label{lem:trace-to-estimate}
For any two distributions $D,D^*$ have total variation distance  $\alpha$,  $ \Pr_{i \sim D^*}( D(i) \leq D^*(i)) \ge \alpha $ 
\end{lem}
\begin{proof}

\ba
0 &\le   \sum_{i: D(i) \leq D^*(i)} D(i) \nnl
 = \sum_{i: D(i) \leq D^*(i)} D(i)-D^*(i)+  \sum_{i: D(i) \leq D^*(i)} D^*(i)   \nle
  \fot\L[\sum_{i: D(i) \leq D^*(i)} (D(i)-D^*(i))+ \sum_{i: D^*(i) \leq D(i)} (D^*(i)-D(i))\R]+ \sum_{i: D(i) \leq D^*(i)} D^*(i)  \nle
 -\alpha +  \Pr_{i \sim D^*}( D(i) \leq D^*(i)) \\ \nn
  \Rightarrow & \Pr_{i \sim D^*}( D(i) \leq D^*(i)) \ge \alpha 
\ea
The third line is true because $\sum_{i: D(i) \leq D^*(i)} (D(i)-D^*(i))= \sum_{i: D^*(i) \leq D(i)} (D^*(i)-D(i))$.
\end{proof}

We think $D^*$ as the correct distribution and $D$ is a distribution far away, with the total variation distance between them being $\alpha=1-\eps$. When we use maximum likelihood estimation to distinguish $D^*$ from $D$, Lemma~\ref{lem:trace-to-estimate} says that the probability of error is less than $\eps$. Now we are ready to prove theorem \ref{thm:pure}.
\begin{proof}
Let $c^*$ be the target concept, and $c$ a concept such that $\Delta(c^*, c) > \eps$.
Recall that we took $$T=\Theta\L(\frac{\log|C|+\log(1/ \delta)} { \eps^2}\R)$$ samples. 
For all $i\in [T]$, apply Theorem \ref{lem:random-measure} to the pair of states $(c^*(x_i), c(x_i))$, we get that with probability $1-\exp(-k d_2/2)$ over random orthonormal measurements $M_i$, 
\ba
\Delta_{tr}\L( M_i(c^*(x_i)), M_i(c(x_i)) \R) > k/\sqrt{2} \Delta_{tr}(c^*(x_i), c(x_i)),
\ea
 where $k$ is a universal constant. Since you can pad some ancilla states to increase $d_2$ without changing trace distances if $\exp(-kd_2/2)$ is not small enough, we ignore this term.  By Chernoff bound, with probability at least $1- 2^{-\Omega(T \eps)}$ over $\{x_i\}$ sampled from $D$, 
\ba
(1/T)\cdot  \sum_i \Delta_{tr}\L( M_i(c^*(x_i)), M_i(c(x_i)) \R) > (1/T)\cdot \sum_i k/\sqrt{2} \Delta_{tr}(c^*(x_i), c(x_i)) \ge \frac{k}{2 \sqrt{2}} \eps. 
\ea
So we can apply Lemma \ref{lem:trace-amp} to get  that with probability at least  $1- 2^{-\Omega(T \eps)}$,
\ba
 \L[ \Delta_{tr}\L(\{ M_i(c^*(x_i))\}, \{M_i(h(x_i))\} \R)\R] \geq 1 - 2^{-\Omega(T \eps^2)}.
\ea
Now, note that by Lemma \ref{lem:trace-to-estimate}, the probability that the maximal likelihood estimation (incorrectly) selects $c$ is at most $(2^{-\Omega(T\eps^2)} + 2^{-\Omega(T\eps)})$. By taking a union bound over all such $c$, we get 
\ba
\Pr[\Delta(c^*,h)>\eps] \leq  (2^{-\Omega(T\eps^2)} + 2^{-\Omega(T\eps)}) \cdot |C| \leq  \delta.
\ea
\end{proof}

\section{PAC Learning Quantum Channels with Mixed State Output}\label{sec:mixed}
The random orthonormal measurement approach in Section~\ref{sec:pure} does not work since two high dimensional mixed states with constant trace distance between them can have negligible Frobenius distance between them. Instead, We follow the intuitions detailed in Section~\ref{sec:intui}. We define the bichromatic state discrimination problem (BSD), solve BSD with PGM techniques , and build our learner algorithm with the BSD solver and a partition sub-algorithm.

  Before we show the algorithms for bicromatic state discrimination, let us first show that we can efficiently amplify the distance between concepts by taking samples. 
\begin{lem} [concept distance amplification]\label{lem:concept-amp-state}
Let $c$ be a quantum channel concept $\eps$-far from the target concept $c^*$. Let $\{x_1,x_2,\dots,x_T \}$ be $T$ inputs  drawn from the distribution $D$.  With probability $1-2^{-\Omega(T \eps)}$ over $\{x_i\}$ drawn, we have 
\ba
  F\L( \bigotimes_{i\in[T]} c(x_i),  \bigotimes_{i\in[T]} c^*(x_i) \R)  \leq 2^{-\Omega(T \eps^2)}
\ea
and
\ba
\Delta_{tr}\L( \bigotimes_{i\in[T]} c(x_i),  \bigotimes_{i\in[T]} c^*(x_i) \R)  \geq 1 - 2^{-\Omega(T \eps^2)}.
\ea
\end{lem}
\begin{proof}
By Chernoff bound, with probability $1-2^{-\Omega(T \eps)}$, 

\ba
\sum_i \Delta_{tr}\L( c(x_i),  c^*(x_i) \R) \geq \frac{1}{2} T \eps.
\ea
Then by Cauchy-Schwarz Inequality, 
\ba
\sum_i (\Delta_{tr}\L( c(x_i),  c^*(x_i) \R))^2 \geq \frac{1}{4} T \eps^2 .
\ea
Then the amplified fidelity is bounded by
\ba
F \L( \bigotimes_{i} c(x_i),  \bigotimes_{i} c^*(x_i) \R) &=\Pi_i F \L( c(x_i),  c^*(x_i) \R) \nnl
\leq \Pi_i \sqrt{1-\L(\Delta_{tr}\L( c(x_i),  c^*(x_i)\R) \R)^2} \nnl
\leq   \exp \L[-\fot \sum_i \L(\Delta_{tr}\L( c(x_i),  c^*(x_i)\R) \R)^2\R] = 2^{-\Omega(T \eps^2)}, 
\ea
where the last inequality is true because  $1-x \leq e^{-x}$. And the amplified trace distance is 
\ba
 \Delta_{tr}\L( \bigotimes_{i\in[T]} c(x_i),  \bigotimes_{i\in[T]} c^*(x_i) \R) \ge 1- F \L( \bigotimes_{i} c(x_i),  \bigotimes_{i} c^*(x_i) \R) =1 - 2^{-\Omega(T \eps^2)}.
\ea
\end{proof}

Lemma  \ref{lem:concept-amp-state} means that we can amplify the distance between tensor products of samples from quantum channels as efficiently as we do on samples of fixed quantum states. This means that PAC learning quantum channels is really similar to approximate state discrimination even in the mixed state case.

Now back to BSD. The bichromatic state discrimination problem (BSD) is defined as follows:
\begin{defn}[Bichromatic State Discrimination Problem (BSD)]
Given complete information of two sets of quantum states, $S_{yes}$ and $S_{no}$, with fidelity $F(S_{yes}, S_{no})\leq \eta$ and size $ S_{yes} \leq N$, $ S_{no} \leq N$, and one copy of an unknown quantum state $\sigma$, the goal is to decide whether $\sigma \in S_{yes}$ or $\sigma \in S_{no}$. We say a quantum algorithm solves $(\eta,N)$-BSD with error $\delta$ if for all $S_{yes}$ and $S_{no}$ such that $F(S_{yes}, S_{no})\leq \eta$, $ S_{yes} \leq N$, and $ S_{no} \leq N$, given complete information about $S_{yes}$ and $S_{no}$ and one copy of an unknown quantum state $\sigma$ as input to the algorithm,  the algorithm output and $yes/no$ answer satisfies the following two conditions:
\begin{enumerate}
    \item  If $\sigma \in S_{yes}$, the learner outputs $yes$ with probability $(1-\delta)$.
    \item If $\sigma \in S_{no}$, the learner outputs $no$ with probability $(1-\delta)$.
\end{enumerate}    
The learner can output anything if $\sigma$ does not come from either of the two sets.
\end{defn}

We show the existence of a BSD solver by first showing that PGM over $S_{yes}\cup S_{no}$ solves the "average case" BSD and then turn it into a "worst case" result by the minimax theorem.

First by slightly modifying a result of~\cite{pgm-bk} and~\cite{pgm-am}, We show that PGM can solve the "average case" BSD:
\begin{lem}[PGM for "average BSD"]\label{lem:bi-pgm}
Let $S_{yes},\, S_{no}$ be two sets of density matrices and $\{p_i\}$ be a probability distribution over $ S_{yes} \cup S_{no} $. \footnote{We will slightly abuse the notation and write $i \in S_{yes}$ or  $j \in S_{no}$ instead of $\sigma_i \in S_{yes}$ or  $\sigma_j \in S_{no}$.} The PGM on $S_{yes} \cup S_{no}, \{p_i\}$ satisfies
\ba
\sum_{i \in S_{yes}}\sum_{j\in S_{no} } \L[ p_i \Pr\L( PGM(\sigma_i)=j\R ) + p_j \Pr( PGM\L(\sigma_j)=i \R)\R] \le  \sum_{i \in S_{yes}}\sum_{j\in S_{no} } F(\sigma_i,\sigma_j).
\ea
\end{lem}
\begin{proof}
See appendix \ref{app:pgm}.
\end{proof}

 We can group together the outputs of the PGM in Lemma~\ref{lem:bi-pgm} and define a binary measurement  $\{E_{yes},E_{no}\}$, where $E_{yes}=\sum_{i\in S_{yes} } E_i$, $E_{no}=\sum_{i\in S_{no} }E_i$, and $\{E_i\}$ is the PGM. By Lemma~\ref{lem:bi-pgm}, the binary measurement solves "average BSD" with error probability at most $\sum_{i \in S_{yes}}\sum_{j\in S_{no} } F(\sigma_i,\sigma_j)$.\footnote{A careful reader might notice that since we only want a binary answer, we are essentially distinguishing the states $A_{yes}=\sum_{i \in S_{yes}} p_i \sigma_i$ and $A_{no}=\sum_{j \in S_{no}} p_j \sigma_j$, and thus the optimal error probability is characterized by trace distance between $A_{yes}$ and $A_{no}$. However, to our knowledge there is no inequality in the literature giving a \emph{lower bound} on trace distance between on linear combinations of density matrices, so actually, the other direction of the trace-distance characterization is the relevant one: Lemma~\ref{lem:bi-pgm} gives a new lower bound on $\Delta_{tr}(A_{yes},A_{no})$.}

 Since the upper bound on error is independent of the distribution $\{p_i\}$, minimax theorem guarantees the existence of a measurement that distinguishes between  $S_{yes}$ and $S_{no}$ for any distribution $\{p_i\}$ with error probability less than $ \sum_{i \in S_{yes}}\sum_{j\in S_{no} } F(\sigma_i,\sigma_j)$\footnote{This argument was used in~\cite{pgm-aram}}. In particular,  if $p_i=1$ for some $\sigma_i\in S_{yes}$, the probability of the minimax measurement mistaking $\sigma_i$ as something in $S_{no}$ is upper bounded by  $ \sum_{i \in S_{yes}}\sum_{j\in S_{no} } F(\sigma_i,\sigma_j)$, and vice versa.  We formalize this discussion as the following Theorem.

\begin{thm}[solver for BSD, Theorem~\ref{thm:bsd} restated] \label{lem:bi}
There exist an algorithm that solves $(\eta,N)$-BSD with error $\delta=N^2 \eta$

\end{thm}
\begin{proof}
Consider the zero sum game between two players where player1 choose a probability distribution $\{p_i\}$ over $S_{yes}\cup S_{no}$ and player2 choose a binary measurement strategy $M$. The score of player1 is given by the following error probability\footnote{We will slightly abuse the notation and write $i \in S_{yes}$ or  $j \in S_{no}$ instead of $\sigma_i \in S_{yes}$ or  $\sigma_j \in S_{no}$.}:
\ba
P_{bi-error} = \sum_{i \in S_{yes}} \L[ p_i \Pr\L( M(\sigma_i)=no\R ) \R] +\sum_{j\in S_{no} } \L[  p_j \Pr \L( M(\sigma_j)=yes \R)\R]  
\ea

 It is easy to check that that strategies of both sides are linear, so we can  apply the minimax theorem to get
 \ba
  \min_{M}  \max_{\{p_i\}} P_{bi-error} =\max_{\{p_i\}}  \min_{M}  P_{bi-error} \leq   \sum_{i \in S_{yes}}\sum_{j\in S_{no} } F(\sigma_i,\sigma_j) \leq N^2\eta,
 \ea 
 where the second inequality is from the promises of $(\eta,N)$ BSD, and the first inequality is shown by considering the binary measurement  $\{E_{yes},E_{no}\}$, where $E_{yes}=\sum_{i\in S_{yes} } E_i$, $E_{no}=\sum_{i\in S_{no} }E_i$, and $\{E_i\}$ is the PGM of Lemma~\ref{lem:bi-pgm}. This means that there is a measurement $M$ whose error probability is less than $\sum_{i \in S_{yes}}\sum_{j\in S_{no} } F(\sigma_i,\sigma_j)$ for all probability distribution $\{p_i\}$. In particular, the error probability is at most $N^2\eta$ when player1 uses the deterministic strategy of always choosing some specific state $\sigma_i \in S_{yes}\cup S_{no}$. Therefore,  algorithm of applying the measurement $M$ solves $(\eta,N)$-BSD with error $N^2\delta$.

%

\end{proof}

Theorem \ref{lem:bi} implies that if we amplify the maximum fidelity between $S_{yes}$ and $S_{no}$  by Lemma \ref{lem:concept-amp-state} to less than $O(1/|C|^2)$, we have a constant error probability in distinguishing whether a state is from $S_{yes}$ or $S_{no}$. By lemma \ref{lem:concept-amp-state} this requires $\Theta \L( \log |C| / \gamma^2 \R)$ samples if the distance between $S_{yes}$ and $S_{no}$ is $\gamma$.

Now we present the partition sub-algorithm.  Let $C_r$ be the set of remaining concepts that have not been cut off by the main algorithm. The sub-algorithm partitions the remaining concepts into three disjoint subsets: $(S_{yes}, S_{unknown},S_{no})$, such that  $|S_{yes}| \geq \frac{1}{9}|C_r|$\footnote{$\frac{1}{9}$ is an arbitrary constant and can be further optimized}, and $\Delta(S_{yes},S_{no})\geq \gamma = \Theta(\eps / \log|C_r|)$. The sub-algorithm might or might not found an extreme case. If no extreme case is found, $|S_{no}| \geq \frac{1}{9}|C_r|$.  If an extreme case  is found, more than $\frac{1}{3}|C_r|$ concepts are $\eps$-close to some concept. The sub-algorithm initialized with every concept in $S_{no}$. It then repeatedly picks a concept $c_c$ from $S_{no}$ and adds concepts within the ball around $c_c$ to $S_{yes}$ and concepts in a $\gamma$-shell around the ball to $S_{unknown}$. The $\gamma$-shell of $S_{no}$ ensures that $\Delta(S_{yes},S_{no}) \geq \gamma$ and we choose the radius of the ball so that the number of concepts added to $S_{yes}$ is greater than half the number of concepts added to $S_{unknown}$ to ensure that $|S_{yes}| > \frac{1}{2}|S_{unknown}|$ in the end. The sub-algorithm keeps adding concepts to $S_{yes}$ and $S_{unknown}$ until $|S_{yes}|+|S_{unknown}| > \frac{1}{3}|C_r|$ or the loop is breaked by an extreme case. The sub-algorithm reports an extreme case if the number of concepts to be added to $S_{yes}$ and $S_{unknown}$ in the current iteration is greater than $\frac{1}{3}|C_r|$. In this case we know that more than $\frac{1}{3}|C_r|$ concepts are around $c_c$. If no extreme case is found, since the loop stops when $|S_{yes}|+|S_{unknown}| > \frac{1}{3}|C_r|$ and the last iteration cannot add more than $\frac{1}{3}|C_r|$ concepts to $S_{yes}$ or $S_{unknown}$, there are at least $(1-\frac{1}{3}-\frac{1}{3})|C_r| > \frac{1}{9}|C_r|$ concepts left in $S_{no}$, and $|S_{yes}| > \frac{1}{3} (|S_{yes}|+|S_{unknown}|) > \frac{1}{9}|C_r|$.

There is another complication in that the distance between the concepts depends on the unknown distribution $D$ and thus cannot be calculated. In stead, we calculate the \emph{empirical distance}  between concepts, $\Delta_{emp}(c_1,c_2)=\frac{1}{T}\sum_{i=1}^T \L [ \Delta_{tr} \L(c_1(x_i),c_2(x_i)\R) \R ] $, which depends on the input points drawn from $D$. We also tune $\eps$ into $\eps/2$ to accommodate for the extra error incurred.

The sub-algorithm is detailed as follows:
\begin{algorithm}\caption{partition sub-algorithm}\label{alg:partition}
 \KwData{ concepts class $C_r$, real number $\eps$.}
 \KwResult{ Set of concepts $S_{yes},S_{unknown},S_{no}$, boolean variable $flag\_extreme$, concept $c_c$}
$S_{no} \leftarrow C_r, \, S_{yes} \leftarrow \emptyset,\, S_{unknown} \leftarrow \emptyset, \,flag\_extreme \leftarrow false,\, \gamma \leftarrow {\eps/(4\log|C_r|) }$.

\While{ $| S_{yes}|+ |S_{unknown}| < \frac{1}{3} |C_r|$\footnote{The loop might also be broken by an extreme case}}{

$c_c \leftarrow $  a random concept in $S_{no}$\;
  Count the number of concept in $S_{no}$ whose distance to $c_c$ is in the interval $\L [(m-1) \gamma, m \gamma\R )$ for all $m \in [1/\gamma]$ and record the number as $b_m$. I.e.  $b_m \leftarrow \L| \{c|\Delta(c,c_c)\in \L [(m-1) \gamma, m \gamma\R ), c \in S_{no}   \}\R|$\;\nllabel{step-bin}

  Find the smallest $i^* \geq2$ such that $ b_{i^*} < 2 \sum_{i \in [i^*-1]} {b_i}  $\;\nllabel{step-yes-unknown-ratio} 
 \If{ $ \sum_{i\in [i^*-1]} {b_i} + b_{i^*} > \frac{1}{3}|C_r|$} { $flag\_extreme \leftarrow true$\; \nllabel{step-extreme}

move everything in $S_{yes}$ and $S_{unknown}$ back to $S_{no}$\;

run line~\ref{step-add-yes-unknown} once\;

Terminate\;}

   For the concepts in $S_{no}$, move the concepts within distance $(i^*-1)\gamma$ of $c_c$ to $S_{yes}$, and move the concepts whose distance to $c_c$ is in $\L [(i^*-1) \gamma, i^* \gamma\R )$ to $S_{unknown}$. I.e. move $ \{c|\Delta(c,c_c)\in \L [0, (i^*-1) \gamma\R ), c \in S_{no}   \}$ to $S_{yes}$ and move $ \{c|\Delta(c,c_c)\in \L [(i^*-1) \gamma, i^* \gamma\R ), c \in S_{no}   \}$ to $S_{unknown}$\;\nllabel{step-add-yes-unknown}

}
\end{algorithm}

\begin{lem}\label{lem:partition}
The output of Algorithm~\ref{alg:partition} satisfies the following conditions:  $(S_{yes},S_{unknown},S_{no})$ is a partition of $C_r$. $\Delta_{emp}(S_{yes},S_{no}) \geq \gamma=\eps/4\log |C_r|$. $|S_{yes}| \geq \frac{1}{9}|C_r|$. If $flag\_extreme=false$, $|S_{no}| \geq \frac{1}{9} |C_r|$. If $flag\_extreme = true$, $ \Delta_{emp}(c,c_c)\leq \eps/2,\, \forall c \in (S_{yes} \cup S_{unknown})$.
\end{lem}
\begin{proof}
First note that in line~\ref{step-yes-unknown-ratio}, $\gamma=\eps/(4 \log|C_r|)$ ensures that $i^*$ exists and $i^*  \leq \eps/(2\gamma)$. This can be proved by contradiction:  if $b_i^* \ge 2 \sum_{i\in [i^*-1]} {b_i}, \forall i^* \leq \eps/(2\gamma)$, then $b_i^*> 2 b_{i^*-1}, \forall i^* \leq \eps/(2 \gamma)$. Together with  $b_1\geq 1$ because $\Delta(c_c,c_c)=0$, we have $b_{\floor{\eps/2\gamma}} \geq 2\cdot 2^{\log |C_r|}b_1 \geq |C_r|$,  a contradiction. 

   $(S_{yes},S_{unknown},S_{no})$ is a partition because it is initialized as a partition and we only moves elements between them. Note that whenever we move something to $S_{yes}$, we move a $\gamma$-thick shell around it to $S_{unknown}$. By triangle inequality of empirical distances between concepts, $\Delta_{emp}(S_{yes},S_{no}) \geq \gamma=\eps/4\log |C_r|$ at the end of every step. 
   
   If no extreme case is found, at each iteration of the loop  at line~\ref{step-add-yes-unknown},  ($ \sum_{i\in [i^*-1]} {b_i})$ concepts are moved to $S_{yes}$ from $S_{no}$, and $  b_{i^*}$ concepts are moved to $S_{unknown}$ from $S_{no}$. Before the last iteration of the loop $|S_{yes}|+|S_{unknown}| \leq \frac{1}{3}|C_r|$, and the number of concepts moved to $S_{yes}$ and $S_{unknown}$ in the last iteration is $ \sum_{i\in [i^*-1]} {b_i} + b_i^* \leq \frac{1}{3}|C_r|$, so  $|S_{no}| \geq (1-\frac{1}{3}-\frac{1}{3})|C_r| > \frac{1}{9}|C_r|$.  Because of the requirement $ b_{i^*} < 2 \sum_{i \in [i^*-1]} {b_i}  $ in line~\ref{step-yes-unknown-ratio}, $ \sum_{i\in [i^*-1]} {b_i} >\frac{1}{3}( \sum_{i\in [i^*-1]} {b_i} + b_{i^*} ) $ and thus $|S_{yes}| >\frac{1}{3}(|S_{yes}|+|S_{unknown}| )$ at the end of every loop. Combined with the loop-termination condition $| S_{yes}|+ |S_{unknown}| > \frac{1}{3} |C_r|$, we have $|S_{yes}| > \frac{1}{9} |C_r|$.  
   
   If an extreme case is found at line~\ref{step-extreme}, because we moved everything back to $S_{no}$, all concepts in $S_{yes}$ or $S_{unknown}$ are added in that one call of line~\ref{step-add-yes-unknown}, and thus they are all $(i^* \gamma)$-close to $c_c$ . Recall that $i^* \gamma \leq \eps/2$,  so everything in $S_{yes}$ or $S_{unknown}$ is $\eps/2$-close to $c_c$. The analysis on $|S_{yes}|$ is a bit subtle. Similar to the previous paragraph, we have $\sum_{i\in [i^*-1]} {b_i} >\frac{1}{3}( \sum_{i\in [i^*-1]} {b_i} + b_{i^*} ) $.  Combined with  $ \sum_{i\in [i^*-1]} {b_i} + b_{i^*} > \frac{1}{3}|C_r|$ to trigger line~\ref{step-extreme}, we have $\sum_{i\in [i^*-1]} {b_i}  > \frac{1}{9}|C_r|$. Since the wiping of $S_{yes}$ and $S_{unknown}$ at the beginning of  line~\ref{step-extreme} only opens more possible concepts to be added to $S_{yes}$, we have $|S_{yes}| \geq \sum_{i\in [i^*-1]} {b_i}  > \frac{1}{9}|C_r|$.
\end{proof}

The main algorithm for mixed state case is then:
\begin{algorithm}
\caption{algorithm for mixed state case}\label{alg:mixed}
\KwData{Concept class $C$, Sampling Oracle ${\cal O}_{c^*, D}$}
\KwResult{hypothesis $h$}
 $C_r \leftarrow C$ \;
 $T \leftarrow \Theta \L(\frac{\log^2 |C|(\log |C|+\log(1/ \delta))} { \eps^2}\R)$
 \;
 
\While{}{

 Call ${\cal O}_{c^*, D}$ $T$ times, getting $T$ samples $\{(x_1,c^*(x_1)),(x_2,c^*(x_2)),\dots\,(x_T,c^*(x_T))\}$\;

 $(S_{yes}, S_{unknown}, S_{no}, flag\_extreme, c_c) \leftarrow$  (Algorithm~\ref{alg:partition})$(C_r,\,\eps)$\;
 

Construct the measurement $M$ in Theorem~\ref{lem:bi} between $S_{yes}$ and $S_{no}$ with the state $\sigma_i$ corresponding to concept $c_i$ being $\sigma_i= \bigotimes_{j\in[T]} c_i(x_j)$\;\nllabel{step-measure}

$Measure\_result \leftarrow  M(\bigotimes_{j\in[T]} c^*(x_j))$.\;

\If{ $Measure\_result$ = $no$}{ remove $S_{yes}$ from $C_r$\;\nllabel{step-no}}

\If{$Measure\_result$ = $yes$ and $flag\_extreme=false$}{ remove $S_{no}$ from $C_r$.\;\nllabel{step-yes}}

\If{ $Measure\_result$ = $yes$ and $flag\_extreme=true$}{ $h \leftarrow c_c$\;\nllabel{step-extreme2-output}

Terminate;\nllabel{step-extreme2}}

}
\end{algorithm}

Now we state and prove our result for mixed state case:
\begin{thm} \label{thm:mixed}
Algorithm \ref{alg:mixed} is a proper $(\eps, \delta)$-PAC learner for any quantum circuit concept class $C$, using $$O\L(\frac{\log^3 |C|(\log |C|+\log(1/ \delta))} { \eps^2}\R)$$ samples.
\end{thm}
\begin{proof}

By Lemma~\ref{lem:partition}, Algorithm~\ref{alg:mixed} removes at least $\frac{1}{9}|C_r|$ concepts from $C_r$ in each loop unless it terminates, so it terminates in  $O(\log {| C |} )$ loops at line~\ref{step-extreme2}. Combined with the fact that Algorithm~\ref{alg:mixed} takes $O\L(\frac{\log^2 |C|(\log |C|+\log(1/ \delta))} { \eps^2}\R)$ samples each loop,  its sample complexity is $$O\L(\frac{\log^3 |C|(\log |C|+\log(1/ \delta))} { \eps^2}\R)$$.

As for the correctness of the algorithm, first note that by Lemma~\ref{lem:partition} the empirical distance between any pair of concepts in $S_{yes}$ and $S_{no}$ is at least $\gamma_0=\eps/(4 \log|C|)$.  


  Consider any pair of concepts $c_i \in S_{yes}$ and $C_j \in S_{no}$, with  the corresponding states $\sigma_i$ and $\sigma_j$. By definition of empirical distance, 
  \ba
  \sum_{k=1}^T \Delta_{tr} \L( c_i(x_k),  c_j(x_k) \R) \geq T\gamma_0
\ea

Then by Cauchy-Schwarz Inequality, 
\ba
\sum_{k=1}^T \Delta_{tr} \L( c_i(x_k),  c_j(x_k) \R)^2 \geq T\gamma_0^2 .
\ea

Then the  fidelity between $\sigma_i$ and $\sigma_j$ is bounded by
\ba \label{eq:f}
F(\sigma_i,\sigma_j) &=
F \L( \bigotimes_{k} c_i(x_k),  \bigotimes_{k} c_j(x_k) \R) \nn \\
&=\Pi_k F \L( c_i(x_k),  c_j(x_k) \R) \nnl
\leq \Pi_k \sqrt{1-\L(\Delta_{tr}\L( c_i(x_k),  c_j(x_k)\R) \R)^2} \nnl
\leq   \exp \L[-\fot \sum_k \L(\Delta_{tr}\L( c_i(x_k),  c_j(x_k)\R) \R)^2\R] \nnl
=2^{-\Omega\L(T \gamma_0^2 \R)}
\ea
where the last inequality is true because  $1-x \leq e^{-x}$.

There are  only two possible ways for Algorithm~\ref{alg:mixed}  to make an error: first is to remove  $c^*$ from $C_r$ in line~\ref{step-no} or line~\ref{step-yes}, and second is to output a far-away concept at line~\ref{step-extreme2} because of the mismatch between empirical distance and true distance.

For the first error, note that $c^*$ always has empirical distance zero to it self, no matter what $\{x_1,x_2,\dots,x_T\}$ are sampled. By Theorem~\ref{lem:bi} and Equation~\ref{eq:f} the error probability in each loop is bounded by
 \ba
 P_{error,1} \leq|C_r|^2 \cdot  2^{-\Omega(T\gamma_0^2)} .
 \ea


Apply union bound over $O(\log {| C |} )$ loop we can bound the total error probability by
\ba
P_{total \, error,1} \leq \log|C||C|^2 \cdot  2^{-\Omega(T\gamma^2)} \leq O\L(\frac{ \delta \, |C|^2 \log |C|}{\poly(|C|)} \R) \leq  O(\delta)
\ea

For the second error, consider a pair of concepts that has distance bigger than $\eps$. By Chernoff bound, the probability that their empirical distance is less than $\fot \eps$ is less than $2^{-\Omega(T\eps^2)}$. Union bound over all $O(|C|^2)$ pairs of concepts, we have
\ba
P_{total \, error,2} \leq |C|^2 \cdot  2^{-\Omega(T\eps^2)} \ll P_{total \, error,1}.
\ea

\end{proof}

\section{Lower Bounds}
In this section we describe two simple lower bounds. One is an  $\Omega((1-\delta)\ln|C|/\eps^2)/\ln(\ln |C|/\eps)$ lower bound on the sample complexity of approximate state discrimination for pure states, which in turn gives lower bounds on the sample complexity of PAC learning quantum channels. The other is an $\Omega(\sqrt{d})$ lower bound on the sample complexity of learning large dimensional \emph{classical   distribution} in the \emph{agnostic} model, which in turn lower bounds approximate state discrimination and PAC learning quantum state in the agnostic model~\cite{Haussler92,KSS94}.

\subsection{Lower Bound for Pure State Case}
\begin{thm}
 The sample complexity of $(\eps,\delta)$-approximate state discrimination on  a set $C$ of pure states is $\Omega((1-\delta)\ln|C|/\eps^2)/\ln(\ln |C|/\eps)$.
\end{thm}
\begin{proof}
This lower bound uses the $\eps$-packing-net construction of \cite{tomography-aram}. In Lemma 5 of the arxiv version of \cite{tomography-aram}, the authors showed the existence of a set $C$ of $d$-dimensional pure states with the following three properties: the distance between each state is at least $\eps$, the Holevo information $\chi_0$ for states uniformly drawn from the set is $O(\eps^2 \ln(d/\eps))$, and  $\ln |C| = \Omega(d)$.  With a simple reduction to communication protocol and Holevo theorem, \cite{tomography-aram} showed that to distinguish states in $C$ with probability $\delta$, $\frac{(1-\delta)\ln|C|-\ln 2}{\chi_0} = \Omega((1-\delta)\ln|C|/\eps^2)/\ln(\ln |C|/\eps)$ samples are required. Since every state in $C$ is $\eps$-far from each other, an $(\eps,\delta)$-approximate  discriminator should be able to distinguish each state in $C$ with probability $\delta$, therefore the discriminator must take $\Omega((1-\delta)\ln|C|/\eps^2)/\ln(\ln |C|/\eps)$ samples. This matches the sample complexity of our pure state algorithm in terms of $\eps$ and $|C|$ with some logarithmic factors.
\end{proof}
\begin{rem}
 Unfortunately, running the same argument with the mixed state $\eps$-packing nets of \cite{tomography-aram} does not give us tighter lower bound, so we don't have a matching lower bound for the mixed state case.
\end{rem}
\begin{cor}
 The proper $(\eps,\delta)$-PAC sample complexity of a concept class $C$ of pure states is $\Omega((1-\delta)\ln|C|/\eps^2)/\ln(\ln |C|/\eps)$.
\end{cor}

\subsection{Agnostic Model} \label{sec:agno}

Agnostic model~\cite{Haussler92,KSS94} is a learning model related to the PAC model. In agnostic model, the target concept does not need to come from the concept class. We formally  define the agnostic model for learning quantum channels as follows:

We consider a learner trying to learn a target concept $c^*$  with respect to an unknown distribution $D$. The learner is also given a concept class $C$. Since the target concept might not be in the concept class $C$, the learner tries the output the concept $c_{opt}$ that minimize the distance to the target concept $c^*$. Here, we consider the concept class $C$ as a finite set of known $d_1$ to $d_2$ dimensional quantum channels, and $D$ as a distribution over the Hilbert space of dimension $d_1$.
Precisely, the learner is given 
access to a sample oracle ${\cal O}_{c^*, D}$ and outputs a hypothesis $h\in C$.
The oracle ${\cal O}_{c^*, D}$ generates i.i.d. samples $(x_i, c^*(x_i))$, where each $x_i \leftarrow D$ is the classical description of a state drawn according to the distribution $D$, and $c^*(x_i)$ is the (potentially mixed) quantum state outputted by $c^*$ on input $x_i$.

The distance between two concepts $c$ and $h$ under the distribution $D$ is the expected trace distance to the target concept $\Delta(c, h)= \mathbb{E}_{x\in D} \L[ \Delta_{tr}(c(x) ,h(x))  \R]$. 
Let $c_{opt}$ be the optimal output, $c_{opt}=\arg \min \L[\Delta(c, c^*)|c\in C  \R].$ The goal of the learner is to find a hypothesis $h \in C$ with  $\Delta(c^*,h) \leq \Delta(c^*,c_{opt}) +\eps$.

 A quantum learning algorithm $A$ is a \emph{
 $(\eps,\delta)$-agnostic learner} for  $C$ if the following holds: For every $c^*$ and distribution $D$, given oracle access to ${\cal O}_{c^*, D}$, $A^{{\cal O}_{c^*, D}}$ outputs an $h\in C$ such that  $\Delta(c^*,h) \leq \Delta(c^*,c_{opt}) +\eps$ with probability at least $1-\delta$. The sample complexity of $A$ is the maximum number of samples $T$ that $A$ needs to query ${\cal O}_{c^*, D}$ to output $h$. The  \emph{ $(\eps,\delta)$-agnostic sample complexity} of a concept class $C$ is the minimum sample complexity over all learners. 

We show that there is no efficient quantum agnostic learner in the following theorem. 
\begin{thm}\label{thm:agno-lower}
For all $\eps<\frac{1}{10}$ and positive integer $d$, there exist a concept class $C$ of dimension 0 to $d$ whose $(\eps,\delta)$-agnostic sample complexity  is $\Omega(\sqrt{d})$.
\end{thm}

\begin{proof}
  We can get the $\Omega(\sqrt{d})$ lower bound with a simple concept class of that only has two concepts. Both of the concepts  are constant channels that output classical distributions. Consider distributions on $d+1$ dimensions $e_0, e_1, \dots, e_d$.  The first concept $C_1$ has all weight on $e_0$. The second concept $C_2$ has weight uniformly distributed over $e_1,\dots ,e_d$. Now consider the following two set of distribution to be learned. $D_1=\{D_{1,i}\}$ has weight $1/3$ on $e_0$ and weight $1/d$ on $2/3$ of dimensions $e_1 \dots,e_d$.  Anything in $D_1$ has distance $2/3$ to $C_1$ and distance $1/3$ to $C_2$, so it should be learned as $C_2$. $D_2=\{D_{2,i}\}$ has weight $1/3$ on $e_0$ and weight $100/d$ on $2/300$ of dimensions $e_1 \dots,e_d$. Anything in $D_2$ has distance $2/3$ to $C_1$ and distance $(1/3+99*2/300+1*(1-2/300))/2 \sim 0.993 $ to $C_2$, so it should be learned as $C_1$. However, $D_1$ and $D_2$ both looks pretty much like a uniform distribution on $e_1,\dots,e_d$. To distinguish them we need to see a collision on $e_1,\dots,e_d$. By a standard birthday bound, we need at least $\Omega(\sqrt{d/100})$ samples to see a collision. Therefore we need $\Omega(\sqrt{d})$ samples to learn classical distributions in agnostic model with constant error. In the regime of $|C|=\poly(d)$, the lower bound means that it's impossible to find an efficient algorithm of sample complexity $O(\polylog |C|)$.
\end{proof}

\begin{rem}
Note that the construction of Theorem~\ref{thm:agno-lower} is based on a classical distribution, so it means that agnostic learning of a classical distribution of many outputs efficiently is also impossible. To the knowledge of the authors, agnostic learning of classical distribution of many output has not been studied in the literature. Also note that classical distribution is not a subclass of pure quantum states, so Theorem~\ref{thm:agno-lower} does not rule out that quantum channel with \emph{pure} state outcomes can be efficiently agnostically learned.
\end{rem}

\begin{rem}
As mentioned in section \ref{sec:indep work}, \cite{buadescu2020improved} studied the problem of quantum hypothesis selection, which can be viewed as a a relax version of agnostic learning, outputting a state $h$ such that $\Delta(c^*,h) \leq 3.01 \Delta(c^*,c_{opt}) +\eps$. 
\end{rem}

\bibliographystyle{alpha}
\bibliography{KM-ref}

\appendix
\section{PGM for "average BSD"} \label{app:pgm}
This section is mostly identical to  Appendix E of~\cite{pgm-am}. \footnote{For those who have read~\cite{pgm-am},~\cite{pgm-am} upper bounds the probability of a state $i$ being mistaken as anything else, and we got our bound on the average probability of things in $S_{yes}$ being mistaken as in $S_{no}$ by changing the size of off-diagonal block matrix $Y$ from $1 \times r-1 $ to $ |S_{yes}| \times |S_{no}| $.}

The main ingredient of the proof is the following lemma (a slight improvement over Lemma 5 in~\cite{pgm-bk}, which lacked a factor of $\fot$). Let $M$ be a positive semidefinite $n \times n$ matrix, symmetrically partitioned as the $2 \times 2$ block matrix $M= \bpm X & Y \\ Y^* & Z \epm$ , where  $X$ is $n_1 \times n_1$, $Y$ is $n_1 \times n_2$ and $Z$ is $n_2 \times n_2$ (with $n=n_1+n_2$. Let $M^2$ be partitioned conformally. Then the off-diagonal blocks of $M$ and $M^2$ satisfy

\bea
\norm{M_{1,2}}^2_2 \leq \frac{1}{2} \norm{(M^2)_{1,2}}_1
 \eea

Note that the validity of this lemma does not extend to general $m\times m$ partitions.

\begin{proof}[Proof of lemma]
We have $M_{1,2}=Y$ and $(M^2)_{1,2}=XY+YZ$. Let us, without loss of generality, assume that $n_1\leq n_2$. From the singular value decomposition of $Y$ we can obtain a basis for representing $M$ in which $Y$ is pseudo-diagonal with non-negative diagonal elements. Let (for $i=1,\dots,n_1$) $x_i$ and $y_i$ be the diagonal elements of $X$ and $Y$, and $z_i$ the first $n_1$ diagonal elements of $Z$, all of which are non-negative. As $M$ is PSD, any of its principal submatrices is PSD too, and we have $y_i\leq \sqrt{x_i z_i} \leq (x_i+z_i)/2$. Thus 
\bea
\norm{Y}^2_2 =\sum^{n_1}_{i=1} y^2_i \leq \frac{1}{2} \sum^{n_1}_{i=1}(XY+YZ)_{i,i}\leq \frac{1}{2}\norm{XY+YZ}_1,
\eea
as required. The last inequality follows from the inequality $ | \tr A| \leq \norm{A}_1$ applied to the square matrix obtained by padding $XY+YZ$ with extra rows containing zero (an operation that does not affect the trace norm).
\end{proof}

To prove Lemma \ref{lem:bi-pgm}, recall that $A_i=p_i \sigma_i$ and let $ s=|S_{yes}| ,\, t= |S_{no}|,\text{and} \,r=s+t$.

 Let $W$ be the $r \times 1$ column block matrix $W:=(A^{1/2}_i A^{-1/4}_0)^r_{j=1}$ . Then $W^* W=\sum^r_{j=1} A^{-1/4}_0A_jA^{-1/4}_0=A^{1/2}_0$. Furthermore, Let $N =W W^*$. Then $N_{i,j}=A^{1/2}_i A_0^{-1/2} A_j^{1/2} $ and $(N^2)_{i,j}= (X X^* X X^*)_{i,j} = A_i^{1/2} A_j^{1/2}$.
   
   We sort the indices so that the first $s$ indices corresponds to $S_{yes}$ and the last $t$ indices correspond to $S_{no}$. We now apply the lemma to the $2\times 2$ block matrix $N=\bpm X &Y \\ Y^* &Z \epm$ where $X$ has $s \times s$ sub-blocks , $Y$ has $s \times t$ sub-blocks, and $Z$ has $t \times t$ sub-blocks, and each sub-block is the matrix $A^{1/2}_i A_0^{-1/2} A_j^{1/2}$. Thus, $N_{1,2}=Y$ is itself a block matrix consisting of the $s \times t$ sub-blocks $A^{1/2}_i A_0^{-1/2} A_j^{1/2}$ for $i \in S_{yes}$ and $j\in S_{no}$. Likewise, $(N^2)_{1,2}$ is a block matrix consisting of the $s \times t$ sub-blocks $A_i^{1/2} A_j^{1/2}$. The lemma then implies, 
   \ba
   \sum_{ i \in S_{yes},j\in S_{no}} \tr A_i A_0^{-1/2} A_j A_0^{-1/2}  &= \sum_{i,j} \norm{A^{1/2}_i A_0^{-1/2} A_j^{1/2}}^2_2=\norm{N_{1,2}}^2_2 \nnl
   \leq \frac{1}{2} \norm{(N^2)_{1,2}}_1 =\frac{1}{2}\norm{(A_i^{1/2} A_j^{1/2})_{ i \in S_{yes},j\in S_{no}}}_1 \nl
   \frac{1}{2}\norm{\sum_{ i \in S_{yes},j\in S_{no}} e_{ij} \otimes (A_i^{1/2} A_j^{1/2})}_1   \nnl
   \leq \frac{1}{2} \sum_{i,j}\norm{A_i^{1/2} A_j^{1/2}}_1 =\frac{1}{2} \sum_{i,j} F(A_i,A_j).
   \ea

The last inequality is just the triangle inequality for the trace norm. Now note that 
$$F(A_i,A_j) =\sqrt{p_i p_j} F(\sigma_i,\sigma_j) \leq F(\sigma_i,\sigma_j)$$ 
and
$$ \tr A_i A_0^{-1/2} A_j A_0^{-1/2} = p_i \Pr\L( PGM(\sigma_i)=j\R ) = p_j \Pr( PGM\L(\sigma_j)=i \R),$$
we have the stated bound on birchromatic state discrimination error probability.

%

%
%
%
%
%
%
%
%
%
%
%
%
%
%
%
%
%
%
%
%
%
%
%

\end{document}